\documentclass{article}
\usepackage[utf8]{inputenc}
\usepackage{enumerate}
\usepackage{xspace}
\usepackage{verbatim}
\usepackage{authblk}
\usepackage{rotating}
\usepackage{hyperref}

\usepackage{amsfonts,amssymb,amsmath,amsthm}
\usepackage{mathtools}
\usepackage{mathrsfs}%

\usepackage{tabularx}
\usepackage{multirow}
\usepackage{hhline}

\usepackage{graphicx}
\usepackage[justification=centering]{caption}

\usepackage[table]{xcolor}

\definecolor{tablegreen}{RGB}{170,251,170}
\definecolor{tablered}{RGB}{255,150,150}
\definecolor{tableyellow}{RGB}{255,255,170}
\newcommand{\greencell}{\cellcolor{tablegreen}}
\newcommand{\redcell}{\cellcolor{tablered}}
\newcommand{\yellowcell}{\cellcolor{tableyellow}}
\definecolor{orange}{RGB}{255,121,0}
\definecolor{yellow}{RGB}{255,220,0}
\definecolor{purple}{RGB}{145,30,255}
\definecolor{pink}{RGB}{255,100,200}
\definecolor{green}{RGB}{80,255,135}
\definecolor{blue}{RGB}{80,200,255}
\definecolor{grey}{RGB}{143,143,143}
\definecolor{red}{RGB}{255,0,0}

\newcommand{\BigO}{\mathcal{O}}
\newcommand{\Tree}{\mathcal{T}}
\newcommand{\Cycle}{\mathcal{C}}
\newcommand{\Wheel}{\mathcal{W}}
\newcommand{\Clique}{\mathcal{K}}
\newcommand{\Path}{\mathcal{P}}
\newcommand{\Star}{\mathcal{S}}
\newcommand{\Graph}{\mathcal{G}}
\newcommand{\ubar}{\overline{u}}
\newcommand{\vbar}{\overline{v}}
\newcommand{\ebar}{\overline{e}}
\newcommand{\sbar}{\overline{s}}
\newcommand{\tbar}{\overline{t}}
\newcommand{\np}{$\mathsf{NP}$}
\newcommand{\p}{$\mathsf{P}$\xspace}
\newcommand{\vne}{\textsc{VNE}\xspace}
\newcommand{\univne}{\textsc{uniVNE}\xspace}
\newcommand{\vnelong}{\textsc{Virtual Network Embedding Problem}\xspace}
\newcommand{\univnelong}{\textsc{Uniform Demand VNE}\xspace}

\newcommand{\flow}{\textsc{SIFP}\xspace}
\newcommand{\flowlong}{\textsc{Minimum Cost Single Integer Flow Problem}\xspace}

\newcommand{\edpp}{\textsc{EDPP}\xspace}

\newcommand{\hgp}{\textsc{HGP}\xspace}
\newcommand{\hgplong}{\textsc{Hamiltonian Graph Problem}\xspace}
\newcommand{\tsp}{\textsc{TSP}\xspace}
\newcommand{\tsplong}{\textsc{Traveling Salesman Problem}\xspace}
\newcommand{\bpp}{\textsc{BPP}\xspace}
\newcommand{\bpplong}{\textsc{Bin Packing Problem}\xspace}
\newcommand{\mcla}{\textsc{MCLA}\xspace}
\newcommand{\mclalong}{\textsc{Minimum Cut Linear Arrangement}\xspace}

\newtheorem{theorem}{Theorem}
\newtheorem{proposition}[theorem]{Proposition}
\newtheorem{lemma}[theorem]{Lemma}

\newtheorem{corollary}[theorem]{Corollary}
\newtheorem{open}[theorem]{Open Problem}

\usepackage{tikz}

\tikzstyle{node}=[
    thick,
    circle,
    inner sep=0.05cm, 
    minimum size=0.5cm,
    draw, 
    fill=white
]

\tikzstyle{vnode}=[
    thick,
    circle,
    inner sep=0.05cm, 
    minimum size=0.5cm,
    draw, 
    fill=gray!35!white
]

\tikzstyle{snode}=[
    rectangle,
    rounded corners=1mm,
    inner sep=0.05cm,
    minimum size=0.5cm,
    draw=black, % Set the border color to black
    line width=0.3mm, % Set the thickness of the border
    fill=gray!35!white
]

\tikzstyle{networkedge}=[thick,color=black,draw]
\tikzstyle{routingedge}=[thick,color=black,draw]
\tikzstyle{placementedge}=[thick,->,dashed,black!50,draw]

\usepackage{environ}

\makeatletter

\newcommand{\problemtitle}[1]{\gdef\@problemtitle{#1}}
\newcommand{\probleminput}[1]{\gdef\@probleminput{#1}}
\newcommand{\problemquestion}[1]{\gdef\@problemquestion{#1}}

\NewEnviron{decproblem}{
    \problemtitle{}\probleminput{}\problemquestion{}
    \BODY
    \par\addvspace{.5\baselineskip}
    \noindent
    \begin{tabularx}{\textwidth}{@{\hspace{\parindent}} l X c}
        \multicolumn{2}{@{\hspace{\parindent}}l}{{\@problemtitle}} \\
        \textit{Input:}    & \@probleminput \\
        \textit{Question:} & \@problemquestion
    \end{tabularx}
    \par\addvspace{.5\baselineskip}
}
\makeatother
\begin{document}

\title{Complexity of the Virtual Network Embedding\\ with uniform demands}

\author[1]{Amal Benhamiche}
\author[2]{Pierre Fouilhoux}
\author[2]{Lucas Létocart}
\author[1]{Nancy Perrot}
\author[1,2]{Alexis Schneider}

\affil[1]{Orange Innovation, F-92320 Châtillon, France}
\affil[2]{Université Sorbonne Paris Nord, CNRS, Laboratoire d'Informatique de Paris Nord, LIPN, F-93430 Villetaneuse, France\footnote{amal.benhamiche@orange.com, pierre.fouilhoux@lipn.univ-paris13.fr, lucas.letocart@lipn.univ-paris13.fr, nancy.perrot@orange.com, alexis.schneider@orange.com}}

\maketitle

\abstract{We study the complexity of the Virtual Network Embedding Problem (VNE), which is the combinatorial core of several telecommunication problems related to the implementation of virtualization technologies, such as Network Slicing.
VNE is to find an optimal assignment of virtual demands to physical resources, encompassing simultaneous placement and routing decisions. 
The problem is known to be strongly NP-hard, even when the virtual network is a uniform path, but is polynomial in some practical cases. This article aims to draw a cohesive frontier between easy and hard instances for VNE.
For this purpose, we consider uniform demands to focus on structural aspects, rather than packing ones. To this end, specific topologies are studied for both virtual and physical networks that arise in practice, such as trees, cycles, wheels and cliques.
Some polynomial greedy or dynamic programming algorithms are proposed, when the physical network is a tree or a cycle, whereas other close cases are shown NP-hard.\\
~\\
{\bf Keywords: } Virtual Network Embedding, Complexity, Graph topologies, Dynamic Programming, Polynomial cases
}

\section{Introduction}

In the near future, network operators will rely on virtualization to deploy on-demand services, such as slices in 5G/6G networks \cite{wesley20205G, wesley2021slicing}, sd-wan overlay topologies \cite{tootaghaj2020sdwan}, or distributed AI topologies \cite{addis2025federatedlearning}. Each service will have a dedicated virtual network, specifically designed to meet its requirements \cite{anderson2005virtualization}. A key challenge related to the deployment of such services is the assignment of virtual network elements to substrate (i.e. physical) network resources. 

The underlying combinatorial core problem is the \vnelong (\vne), that requires simultaneous allocation and routing decisions and can be introduced as follows.

In the virtual network, nodes correspond to virtual network functions and containers with computational demands (e.g. CPU, RAM or storage), and are connected by edges having bandwidth requirements.
Nodes of the substrate network proposes computational resources while edges have a fixed bandwidth capacity.
Virtual nodes must then be placed on substrate nodes and virtual edges routed on substrate paths linking those nodes.
Both substrate and virtual networks might have specific structures encountered in real cases, such as rings for optical fiber networks \cite{cosares1994sonet} or trees for data centers \cite{ballani2011datacenter}. \\

Solving the \vne problem and its variants is well-studied in the literature (a survey is done by \cite{fischer2013survey}). %\cite{rost2019approx} for an approximation, \cite{yan2020rl} for a reinforcement learning approach).
In contrast, underlying structural aspects are less considered.
In particular, few articles deal with \vne complexity. 
The first valid proof of \np-hardness for the problem was proposed by \cite{tieves2016complexity}, correcting an earlier reference.
These results were strengthen by \cite{rost2020hardness}, where \vne is proven to remain \np-hard, when considering capacity as well as latency and placement restrictions constraints. 
\vne is shown polynomial when the virtual network is a star with uniform demands on nodes and edges by \cite{rost2015stars}.
For a uniform-demand path virtual network, however, the problem is known to be already \np-hard \cite{wu2020path}. 
Finally, the complexity of \vne for several tree topologies on virtual and substrate networks, when both networks have equal size, is studied by \cite{pankratov2023tree}. 

In this paper, we investigate the natural question of drawing a clearer boundary between \np-hard cases and polynomial ones. This allows for a better understanding of the inherent complexity of the problem.
To this end, we consider the specific case of \univnelong (\univne). Indeed, when the demands are non-uniform, the virtual node assignment contains Bin Packing or Knapsack aspects (which are the key of some proofs of \cite{tieves2016complexity, pankratov2023tree}), leading to \np-hardness for even simple topologies. Note that \univne retains the simultaneous placement and routing aspects of the initial problem.  \\

The \univne can be formally defined as follows. 
The virtual network (resp. substrate) is an undirected simple connected graph $\Graph_r = (V_r, E_r)$  (resp. $\Graph_s = (V_s, E_s)$) with $n_r$ nodes (resp. $n_s$ nodes). Substrate node (resp. edge) have a computational (resp. bandwidth) capacity $c(u) \in \mathbb{N}^+$ for $u \in V_s$ (resp. $c(e) \in \mathbb{N}^+$ for $e \in E_s$).

Allocating the substrate resources to the virtual demands is to find a \textit{mapping} $m = (m_V, m_E)$ of $\Graph_r$ on $\Graph_s$, which is a pair of functions, such that $m_V: V_r \rightarrow V_s$ is called the \textit{node placement}; and $m_E: E_r \rightarrow P_s$, is called the \textit{edge routing}, where $P_s$ is the set of loop-free paths of $\Graph_s$ and for each $\ebar = (\ubar,\vbar) \in E_r$, $m_E(\ebar)$ is a path of $\Graph_s$ connecting $m_V(\ubar)$ and $m_V(\vbar)$.

A substrate node (resp. edge) is said to \textit{host} the virtual nodes that are placed on it (resp. virtual edges that are routed on it). A mapping (also called an embedding) is said \textit{feasible} when the capacity constraints are satisfied, i.e. the number of virtual nodes (resp. edges) being hosted by a substrate node (resp. edge) is less or equal than the capacity of this node (resp. edge).

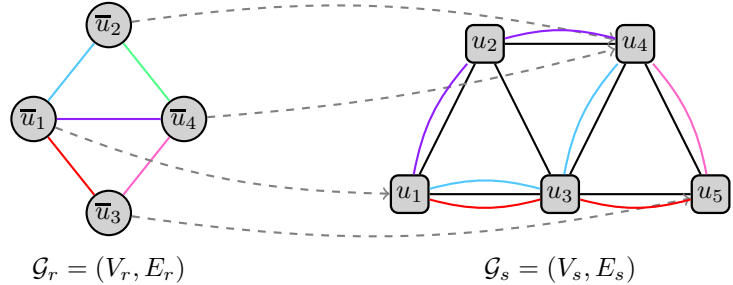
\begin{figure}
    \centering
    \begin{tikzpicture}
    % virtual network
    \begin{scope}
        \node at (2,-2) {$\Graph_r = (V_r, E_r)$};
        \node[vnode] (v1) at (1,0){$\ubar_1$};
        \node[vnode] (v2) at (2,1.25){$\ubar_2$};
        \node[vnode] (v3) at (2,-1.25){$\ubar_3$};
        \node[vnode] (v4) at (3,0){$\ubar_4$};
    \end{scope}
    \draw[networkedge, blue] (v1) -- (v2);
    \draw[networkedge, purple] (v1) -- (v4);
    \draw[networkedge, green] (v2) -- (v4);
    \draw[networkedge, red] (v1) -- (v3);
    \draw[networkedge, pink] (v3) -- (v4);

    % substrate network
    \begin{scope}[xshift=6cm]
        \node at (2,-2) {$\Graph_s = (V_s, E_s)$};
        \node[snode] (u1) at (0,-1){$u_1$};
        \node[snode] (u2) at (1,1){$u_2$};
        \node[snode] (u3) at (2,-1){$u_3$};
        \node[snode] (u4) at (3,1){$u_4$};
        \node[snode] (u5) at (4,-1){$u_5$};
    \end{scope}
    \draw[networkedge] (u1) -- (u2);
    \draw[networkedge] (u1) -- (u3);
    \draw[networkedge] (u2) -- (u4);
    \draw[networkedge] (u2) -- (u3);
    \draw[networkedge] (u3) -- (u4);
    \draw[networkedge] (u4) -- (u5);
    \draw[networkedge] (u3) -- (u5);

    % placement(in 2nd)
    \draw[placementedge] (v1) to[bend right=10] (u1);
    \draw[placementedge] (v2) to[bend left = 10] (u4);
    \draw[placementedge] (v3) to[bend right= 10] (u5);
    \draw[placementedge] (v4) to[bend right = 5] (u4);

    % routing (in 3rd and 4th)
    \draw[networkedge, blue] (u4) to[bend right=15] (u3);
    \draw[networkedge, blue] (u3) to[bend right=15] (u1);
    \draw[networkedge, pink] (u4) to[bend left=15] (u5);
    \draw[networkedge, red] (u1) to[bend right=15] (u3);
    \draw[networkedge, red] (u3) to[bend right=15] (u5);
    \draw[networkedge, purple] (u4) to[bend right=15] (u2);
    \draw[networkedge, purple] (u2) to[bend right=15] (u1);
    %\draw[networkedge, green] (u4) edge[out=10, in=70, looseness=5, loop] ();
        
    \end{tikzpicture}
    \caption{Example of a mapping}
    \label{fig:section1:example-vne}
\end{figure}

Figure \ref{fig:section1:example-vne} illustrates the embedding of a simple virtual graph on the left, over a substrate graph on the right. The dotted arrows show the placement of virtual nodes. A virtual edge of a given color is routed using a substrate path of the same color.
Note that one substrate node (resp. edge), as $u_4$ (resp. $(u_1, u_3)$), can host several virtual nodes (edges). 
The edge $(\ubar_2, \ubar_4)$ is then routed using an empty path.

\begin{decproblem}
    \problemtitle{Existence-\univne (E-\univne)}
    \probleminput{$\Graph_r = (V_r, E_r)$, $\Graph_s = (V_s, E_s)$ with capacities $c: V_s \times E_s \rightarrow \mathbb{N}^+$.}
    \problemquestion{Is there a feasible mapping of $\Graph_r$ on $\Graph_s$ ?}
\end{decproblem}

Assigning a virtual node (resp. edge) on a substrate node $u \in V_s$ (resp. edge $e \in E_s$) infers a cost $w_u \in \mathbb{N}_+$ (resp. $w_e \in \mathbb{N}_+$).
The cost of a mapping is then the sum of the placement and routing costs. The Cost-\univne (C-\univne) asks to find the minimum cost embedding of $\Graph_r$ on $\Graph_s$, and is equivalent to the decision problem:

\begin{decproblem}
    \problemtitle{Cost-\univne (C-\univne)}
    \probleminput{$\Graph_r = (V_r, E_r)$, $\Graph_s = (V_s, E_s)$ with capacities $c: V \times E \rightarrow \mathbb{N}_+$, and costs $w: V \times E \rightarrow \mathbb{N}_+$ and a budget $\Gamma \in \mathbb{N}$.}
    \problemquestion{Is there a feasible mapping of $\Graph_r$ on $\Graph_s$ of cost at most $\Gamma$ ?}
\end{decproblem}

As mentioned, E-\univne is known to be \np-complete even when the virtual graph is a path. In this article, we will consider particular cases on both the virtual and substrate graphs. The considered topologies are trees (denoted by $\Tree$), paths ($\Path$), wheel ($\Wheel$), cycles ($\Cycle$), stars ($\Star$) and cliques ($\Clique$). 

We propose a taxonomy $\alpha$-$\langle \beta_r \rightarrow \gamma_s \rangle$ corresponding to \univne cases where $\alpha$ is the variant (E or C), $\beta$ is the virtual graph topology, and $\gamma$ is the substrate graph topology.
For example, E-$\langle \Path_r \rightarrow \Cycle_s \rangle$ corresponds to the case of E-\univne with a virtual path and a substrate cycle.
Note that, for a $n$ nodes path or cycle, the nodes are numbered $u_1, \ldots, u_n$, according to their position. The same is true for a wheel of $n+1$ nodes, with $u_c$ denoting the center node.
Additionally, for a cycle (and for the outer cycle of a wheel), an index denoted by $\in\{1,\ldots,n\}$ will be given modulo $n$, e.g. $u_{n+1}$ is equal to $u_1$. \\

The paper is structured as follows. Each section corresponds to a substrate graph topology, and each subsection to one or several virtual graph topologies. 
Section 2 deals with substrate general graphs, which have already been studied for virtual stars and paths. We propose further \np-complete results for other virtual topologies.
In Section 3, we initiate the study of substrate cliques, which are close to the core of telecommunication networks.
Sections 4 and 5 focus on substrate trees and cycles, which are realistic topologies for data-centers and fiber networks. Beyond the difficult cases, we propose several polynomial cases using dynamic and greedy algorithms.
In conclusion we give a summary of the results and open questions in Table~\ref{table-complexity}. 

\section{Substrate general graph}

To the best of our knowledge, only two results can be found in the literature regarding the complexity of \univne. The C-$\langle \Star_r \rightarrow \Graph_s \rangle$ is proven to be in $\mathcal{P}$ by \cite{rost2015stars}, using integer flow algorithms. 
However E-$\langle  \Path_r \rightarrow \Graph_s \rangle$ is proven to be $\mathcal{NP}$-hard by \cite{wu2020path}, implying that \univne is \np-complete in the general case. We will now show that other basic topologies are also \np-complete to embed.

\subsection{Virtual cycle and wheel}

\paragraph{Virtual cycle}
Consider a graph $\Graph = (V, E)$ with $n$ nodes $u_1, \ldots, u_n$. A Hamiltonian cycle of $\Graph$ is a cycle that visits every node exactly once, denoted $u_{\pi(1)}, \ldots, u_{\pi(n)}$, where $\pi$ is a permutation of $\{1, \ldots, n\}$. The associated decision problem is defined as follows:

\begin{decproblem}
    \problemtitle{\hgplong (\hgp)}
    \probleminput{An undirected graph $\Graph = (V, E)$.}
    \problemquestion{Does $\Graph$ contain a Hamiltonian cycle ? }
\end{decproblem}

\noindent The \hgp is known to be $\mathcal{NP}$-hard, even when the node degree is up to three \cite{garey1974problems}.
\cite{wu2020path} show that $\langle \Path_r \rightarrow \Graph_s \rangle$ is \np-complete by using a reduction from the Supereulerian Graph Problem, which is equivalent to HGP for graphs with node degree up to three \cite{pulleyblankSGP1979}. We directly use the HGP to obtain new complexity results for cycles and wheels virtual graphs.

\vspace{0.25cm}
\begin{theorem}\label{E-CrGs} \label{theorem:section2:cyclewheel:E-CrGs}
    \textrm{E-}$\langle  \Cycle_r \rightarrow \Graph_s \rangle$ is $\mathcal{NP}$-hard
\end{theorem}

\begin{proof}
    Consider a HGP instance $\Graph$ with $n \ge 2$ nodes of degree two or three. An instance of E-\univne can be constructed as follows. The virtual graph is a cycle graph $\Cycle_r$ with $n$ nodes $\ubar_1, \ldots, \ubar_n$. The substrate graph $\Graph_s$ is $\Graph$, with unit capacities on nodes and edges. We show that $\Graph$ contains a Hamiltonian cycle if and only if there is a feasible mapping of $\Cycle_r$ on $\Graph_s$.

    ($\Rightarrow$) Suppose $\Graph$ contains a Hamiltonian cycle $u_{\pi(1)}, \ldots, u_{\pi(n)}$. The following mapping $m = (m_V, m_E)$ is clearly feasible:
    $m_V(\ubar_i) = u_{\pi(i)}$ and $m_E( \ubar_i, \ubar_{i+1}) = (u_{\pi(i)}, u_{\pi(i+1)})$, for $i \in \{1, \ldots, n\}$.

    ($\Leftarrow$) Suppose there is a feasible mapping $m = (m_V, m_E)$. Since the $n$ substrate nodes have unit capacities, every substrate node hosts exactly one of the $n$ virtual nodes. Let us show that a virtual edge is routed on a single substrate edge. Indeed, for a substrate node $u = m_V(\ubar_i)$, $i \in \{1, \ldots, n\}$, two of its adjacent edges belongs to paths $m_E(\ubar_{i-1}, \ubar_i)$ and $m_E(\ubar_{i}, \ubar_{i+1})$. Since $u$ has degree up to three, no other virtual edge can be routed through $u$. Moreover, as a mapping preserves the connectivity of $\Cycle_r$, the sequence $m_V(\ubar_1),  m_V(\ubar_2), \ldots, m_V(\ubar_n), m_V(\ubar_1) $ is thus a Hamiltonian cycle of $\Graph$. 
\end{proof}

\paragraph{Virtual wheel} 

The following result is a not so straightforward extension of the previous theorem.

\vspace{0.25cm}
\begin{corollary}\label{corollary:section2:cyclewheel:E-WrGs}
    E-$\langle  \Wheel_r \rightarrow \Graph_s \rangle$ is $\mathcal{NP}$-hard
\end{corollary}

\begin{proof}
    Consider a HGP instance $\Graph$ with $n \ge 4$ nodes of degree at most three.
    An instance of E-\univne can be constructed as follows. The virtual graph is a wheel graph $\Wheel_r$, with nodes $\ubar_1, \ldots, \ubar_n, \ubar_{c}$. The substrate graph $\Graph_s$ is $\Graph$ plus an additional node $u_{c}$ connected to every node, with unit capacities on nodes and edges. 
    We show that $\Graph$ contains a Hamiltonian  cycle if and only if there is a feasible mapping of $\Wheel_r$ on $\Graph_s$.

    ($\Rightarrow$) Suppose $\Graph$ contains a Hamiltonian cycle $u_{\pi(1)}, \ldots, u_{\pi(n)}$. The following mapping $m = (m_V, m_E)$ is clearly feasible: $m_V(\ubar_{c}) = u_{c}$ and \\
    \indent $m_V(\ubar_i) = u_{\pi(i)},~ m_E(\ubar_i, \ubar_{i+1}) =\{(u_{\pi(i)}, u_{\pi(i + 1)})\}, ~ m_E(\ubar_i, \ubar_{c}) = \{(u_{\pi(i)}, u_{c})\}$, ~ for $i \in \{1, \ldots, n\}.$

    ($\Leftarrow$) Suppose that there is a feasible mapping $m = (m_V, m_E)$. Since only $u_{c}$ has $n$ adjacent edges of unit capacities, this substrate node is the only one can that can host $\ubar_c$. Moreover, edges $\{(u_{\pi(i)}, u_{c})\}$, $i\in\{1,\ldots,n\}$, belong to paths $m_E(\ubar_i, \ubar_{c})$, then the only edges of $\Graph_s$ that are available for routing the edges of the cycle of $\Wheel_r$ are those of $\Graph$. Therefore, as shown in the proof of \ref{E-CrGs} the sequence $m_V(\ubar_1), m_V(\ubar_2), \ldots, m_V(\ubar_n)$ is a Hamiltonian cycle of $\Graph$.
\end{proof}

\subsection{Virtual clique}

The case E-$\langle \Clique_r \rightarrow \Graph_s \rangle$ has been introduced inside a proof by \cite{tieves2016complexity}  to show the NP-hardness of E-\univne. We propose a counter-example for that proof together with a new one.
This reduction uses the Hadwinger Number Problem, which is a NP-hard \cite{eppstein2009clique} problem related to minors. A graph $\mathcal{H}$ is a minor of a graph $\Graph$ if $\mathcal{H}$ can be obtained from $\Graph$ by deleting edges, vertices and by contracting edges.

\begin{decproblem}
    \problemtitle{\textsc{Hadwinger Number Problem (HNP)}}
    \probleminput{An undirected graph $\Graph = (V,E)$ and an integer $h$.}
    \problemquestion{Does $\Graph$ has a clique of size $h$ as a minor ? }
\end{decproblem}

\noindent Given an instance $(\Graph, h)$ of HNP, an instance of E-\univne is constructed by \cite{tieves2016complexity} as follows. The virtual graph $\Clique_r$ is a clique of size h. The substrate graph $\Graph_s$ is $\Graph$, with unit capacities on nodes and edges. The authors claim that there exists a feasible mapping of $\Graph_r$ on $\Graph_s$ if and only if a clique of size $h$ is a minor of $\Graph$. However, a feasible mapping of $\Clique_4$ can exist even if $\Graph_s$ is a series-parallel graph, as Figure~\ref{figure:section2:clique:counter-example} depicts. Since series-parallel graphs are the graphs that do not contain $\Clique_4$ as a minor, this contradicts the claim.

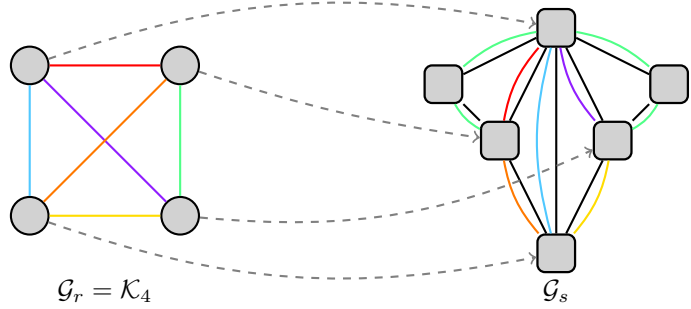
\begin{figure}
\centering
\begin{tikzpicture}
    % virtual graph
    \node at (1, -2) {$\Graph_r = \Clique_4$};
    \node[vnode] (v1) at (0,1){};
    \node[vnode] (v2) at (2,1){};
    \node[vnode] (v3) at (0,-1){};
    \node[vnode] (v4) at (2,-1){};
    \draw[networkedge, red] (v1) -- (v2);
    \draw[networkedge, purple] (v1) -- (v4);
    \draw[networkedge, green] (v2) -- (v4);
    \draw[networkedge, blue] (v1) -- (v3);
    \draw[networkedge, yellow] (v3) -- (v4);
    \draw[networkedge, orange] (v3) -- (v2);

    % substrate network
    \begin{scope}[xshift=7cm] % Adjust the value of xshift as needed
        \node at (0, -2) {$\Graph_s$};
        \node[snode] (u1) at (0, 1.5){};
        \node[snode] (u2) at (0, -1.5){};
        \node[snode] (u3) at (-1.5, 0.75){};
        \node[snode] (u4) at (1.5, 0.75){};
        \node[snode] (u5) at (-0.75, 0){};
        \node[snode] (u6) at (0.75, 0){};
        \draw[networkedge] (u1) -- (u2);
        \draw[networkedge] (u1) -- (u3);
        \draw[networkedge] (u3) -- (u5);
        \draw[networkedge] (u5) -- (u1);
        \draw[networkedge] (u5) -- (u2);
        \draw[networkedge] (u1) -- (u6);
        \draw[networkedge] (u1) -- (u4);
        \draw[networkedge] (u6) -- (u4);
        \draw[networkedge] (u6) -- (u2);
    \end{scope}

    % placement
    \draw[placementedge] (v1) to[bend left=15] (u1);
    \draw[placementedge] (v3) to[bend right=15] (u2);
    \draw[placementedge] (v2) to[bend right=5] (u5);
    \draw[placementedge] (v4) to[bend right=15] (u6);

    % routing
    
    \draw[networkedge, blue] (u1)[bend right = 15] to (u2);
    \draw[networkedge, red] (u1) to[bend right = 15] (u5);
    \draw[networkedge, purple] (u1) to[bend right = 15] (u6);
    \draw[networkedge, green] (u5) to[bend left = 15] (u3);
    \draw[networkedge, green] (u3) to[bend left = 15] (u1);
    \draw[networkedge, green] (u1) to[bend left = 15] (u4);
    \draw[networkedge, green] (u4) to[bend left = 15] (u6);
    \draw[networkedge, yellow] (u6) to[bend left = 15] (u2);
    \draw[networkedge, orange] (u5) to[bend right = 15] (u2);
    
\end{tikzpicture}
\caption{Counter example for the reduction from HNP to E-$\langle \Clique_r \rightarrow \Graph_s \rangle$ by \cite{tieves2016complexity} }
\label{figure:section2:clique:counter-example}
\end{figure}

Consider now the Edge-Disjoint Path Problem, which is $\mathcal{NP}$-hard, even for Series Parallel Graph \cite{nishizeki2001edpp}:

\begin{decproblem}
    \problemtitle{\textsc{Edge-Disjoint Path Problem (EDPP)}}
    \probleminput{An undirected graph $\Graph$ and a set of $k$ pairs of distinct \textit{terminal} nodes $(s_i, t_i)_{i \in \{1, \ldots, k\}}$~.}
    \problemquestion{Are there $k$ edge-disjoint paths $P_1, \ldots, P_k$, such that $P_i$ connects $s_i$ and $t_i$, $i \in \{1, \ldots, k\}$~? }
\end{decproblem}

\begin{theorem}\label{theorem:section2:cyclewheel:E-KrGs}
    E-$\langle  \Clique_r \rightarrow \Graph_s \rangle$ is \np-complete
\end{theorem}

\begin{proof}
    Consider an EDPP instance $(\Graph, (s_i, t_i)_{i \in \{1, \ldots, k\}})$. An instance of E-\univne can be constructed as follows.    
    The virtual graph is a clique graph $\Clique_r$ with $2k$ nodes $\overline{s}_1, \overline{t}_1, \ldots, \overline{s}_k, \overline{t}_k$.
    The substrate graph $\Graph_s$ is $\Graph$, plus an additional edge between every terminal nodes not in the same pair (if this edge is not already in $\Graph$).
    A substrate node has a capacity of one for terminal nodes, and zero otherwise. Substrate edges have unit capacities, plus one if the edge is both in $\Graph$ and links two terminal nodes not in the same pair.
    We show that there are $k$ edge-disjoint paths in $\Graph$ if and only if there is a feasible mapping of $\Clique_r$ on $\Graph_s$. 

    ($\Rightarrow$) Suppose there are $k$ disjoint paths $P_1, \ldots, P_k$ between the pairs of terminal nodes $(s_i, t_i)_{i \in \{1, \ldots, k\}}$. The following mapping $m = (m_V, m_E)$ respects node and edges capacities, and is feasible: \\
    \indent $m_V(\sbar_i) = s_i$,~ $m_V(\tbar_i) = t_i$,~ $m_E(\sbar_i, \tbar_i) = P_i$, ~ for $i \in \{1, \ldots, k\}$~ and \\
    \indent $m_E(\sbar_i, \sbar_j) = \{(s_i, s_j)\}$, ~$m_E(\sbar_i, \tbar_j) = \{(s_i, t_j)\}$,~ $m_E(\tbar_i, \tbar_j) = \{(t_i, t_j)\}$, ~ for $i \neq j \in \{1,\ldots,  k\}$.

    ($\Leftarrow$) Suppose that there is a feasible mapping $m = (m_V, m_E)$. Since every node is interchangeable in a uniform-demand virtual clique, we assume, w.l.o.g., that $m_V(\sbar_l) = s_l$ and $m_V(\tbar_l) = t_l$, for $l \in \{1, \ldots k\}$. \\
    First consider that $(\sbar_l, \sbar_{l'})$, $(\sbar_l, \tbar_{l'})$ and $(\tbar_l, \tbar_{l'})$, for $l \neq l' \in \{1, \ldots, k\}$, are \textit{correctly} routed, that is: $m_E(\sbar_l, \sbar_{l'}) = \{ (s_l, s_{l'}) \}$, $m_E(\sbar_l, \tbar_{l'}) = \{ (s_l, t_{l'}) \}$ and  $m_E(\tbar_l, \tbar_{l'}) = \{ (t_l, t_{l'}) \}$. Note that such partial routing lets exactly one unit of capacity on edges that exists in $\Graph$ (and none on additional edges). The paths $m_E(\sbar_l, \tbar_l)$, $l \in \{1, \ldots, k\}$, are then edge disjoint on $\Graph$, and are thus a solution of EDPP. \\
    On the contrary, let us consider there is one \textit{incorrectly} routed virtual edge, e.g. $(\sbar_{l}, \tbar_{l'})$, $l \neq l' \in \{1, \ldots, k\}$.  We will construct another mapping $m' = (m_V, m'_E)$, where $m'_E(\sbar_{l}, \tbar_{l'}) = \{(s_l, t_{l'})\}$. There might be an edge, say $(\sbar_{i}, \tbar_j)$, $i, j \in \{1, \ldots, k\}$, that is routed through $(s_l, t_{l'})$ in $m_E$. Then $m_E(\sbar_{i}, \tbar_j) = \mu_1 \cup \{(s_l, t_{l'}) \} \cup \mu_2$ where $\mu_1$ and $\mu_2$ are distinct paths in $\Graph_s$. We set $m'_E(\sbar_{i}, \tbar_j) = \mu_1 \cup m_E(\sbar_l, \tbar_{l'}) \cup \mu_2$. For all the other virtual edges, $m'_E$ is equal to $m_E$.
    It can be seen that $m'$ is feasible, and correctly routes virtual edges $(\sbar_l, \sbar_{l'})$, $(\sbar_l, \tbar_{l'})$ and $(\tbar_l, \tbar_{l'})$, for $l \neq l' \in \{1, \ldots, k\}$. Thus, as shown previously, the paths $m_E(\sbar_i, \tbar_i)$, $i \in \{1, \ldots, k\}$, are a solution of EDPP. 
    If there is more than one incorrectly routed edge, this process can be reapplied.
\end{proof}

\section{Substrate Clique}

First, remark that when the substrate graph is a clique $\Clique_s$ such that there is a $n_r \le n_s$, a feasible mapping of $\Graph_r$ on $\Clique_s$ trivially exists.
However, we will show that, when $n_r > n_s$, E-$\langle \Graph_r \rightarrow \Clique_s \rangle$ is \np-complete, and that the cost variant is also \np-complete even when $n_r \le n_s$.

\subsection{Virtual cycle and path}

For a virtual cycle or path, the previous result can be extended to the case where $n_r > n_s$:

\vspace{0.25cm}
\begin{proposition}\label{proposition:section3:cyclepath:hamilto}
    E-$\langle \Cycle_r \rightarrow \Graph_s \rangle$ and E-$\langle \Path_r \rightarrow \Graph_s \rangle$ are in \p and have a trivial solution when $\Graph_s$ is Hamiltonian.
\end{proposition}
 
\begin{proof}
    Since $\Graph_s$ is Hamiltonian, it contains a Hamiltonian cycle $(u_{\pi(1)},\ldots, u_{\pi(n_s)})$. 
    A feasible mapping $m = (m_V, m_E)$ can be constructed with the following procedure: set $m_V(\ubar_i) = u_{\pi(1)}$, for $i \in \{1, \ldots, c(u_{\pi(1)})\}$; $m_E(\ubar_{c(u_{\pi(i)})}, \ubar_{c(u_{\pi(i)})+1}) = \emptyset$, for $i \in \{1, \ldots, c(u_{\pi(1)}) - 1 \}$; $m_E(\ubar_{c(u_{\pi(1)})}, \ubar_{c(u_{\pi(1)})+1}) = ( (u_1, u_2) )$. Repeat this for $u_2, u_3, \ldots, u_{\pi(j)}$, such that $\sum_{l \in \{1, \ldots, j\} } c(u_{\pi(l)}) \ge n_r$, i.e. until all virtual nodes are placed. If the virtual graph is a cycle, set $m_E(\ubar_{n_r}, \ubar_1) = ( (u_{\pi(j)}, u_{\pi(j)+1}), \ldots, (u_{\pi(n_s)}, u_{\pi(1)}) )$. 
    Since we suppose $n_r \le \sum_{u \in V_s} c(u)$, the procedure always terminates. 
\end{proof}

This clearly shows the polynomiality of E-$\langle \Cycle_r \rightarrow \Clique_s \rangle$ and E-$\langle \Path_r \rightarrow \Clique_s \rangle$, which is not the case for the cost variant. Consider the \textsc{Traveling Salesman Problem}, which is \np-complete \cite{garey1974problems}:

\begin{decproblem}
    \problemtitle{\tsplong (\tsp)}
    \probleminput{An undirected clique graph $\Clique = (V, E)$ with $n$ nodes, a distance function $d(u_i, u_j) \in \mathbb{N}^*$ for $u_i \neq u_j \in V$, and an integer $\Lambda \in \mathbb{N}$~.}
    \problemquestion{Does $\Clique$ contain a Hamiltonian cycle $u_{\pi(1)}, \ldots, u_{\pi(n)}$ of length, i.e. $\sum_{i \in \{1, \ldots, n-1\}} d(u_{\pi(i)}, u_{\pi(i+1)})$, at most $\Lambda$?}
\end{decproblem}

\begin{theorem}\label{theorem:section3:cyclepath:cost}
    C-$\langle \Cycle_r \rightarrow \Clique_s \rangle$ and C-$\langle \Path_r \rightarrow \Clique_s \rangle$ are \np-complete
\end{theorem}

\begin{proof}
    Consider a \tsp instance $\Clique, d, \Lambda$. An instance of C-\univne can be constructed as follows.
    The virtual graph is a cycle graph $\Cycle_r$ with $n$ nodes.
    The substrate graph is a clique graph $\Clique_s$ with $n$ nodes. It has unit capacity on every node and edge. The costs are: $w(u_i) = 0$, for $i \in \{1, \ldots, n\}$, and $w(u_i, u_j) = \Lambda + d(c_i, c_j)$, for $i, j \in \{1, \ldots, n\}$, $i \neq j$. The bound is $\Gamma = (n + 1)\Lambda $.
    We show that $\Clique$ contains a Hamiltonian cycle of length at most $\Lambda$ if and only if there is a feasible mapping of $\Cycle_r$ on $\Clique_s$ of cost at most $\Gamma$.

    ($\Rightarrow$) Suppose $\Clique$ contains a Hamiltonian cycle $u_{\pi(1)}, \ldots, u_{\pi(n)}$ of length $\Lambda' \le \Lambda$. The following mapping $m = (m_V, m_E)$ is clearly feasible, and $W(m) = n \Lambda + \Lambda' < (n+1)\Lambda = \Gamma$: \\
    $m_V(\ubar_i) = u_{\pi(i)}$ and $m_E(\ubar_i, \ubar_{i+1}) = \{(u_{\pi(i)}, u_{\pi(i+1)})\}$, for $i \in \{1, \ldots, n\}$  

    ($\Leftarrow$) Suppose that there is a feasible mapping $m = (m_V, m_E)$ of cost $\Gamma' \le \Gamma$.
    Note that a virtual edge is routed on exactly a single substrate edge: indeed nodes have unit capacities, each path $m_E(\ubar_i, \ubar_{i+1})$, $i \in \{1, \ldots, n\}$ uses at least one edge, thus $\Gamma' > n \Lambda$. Moreover if a path of $m_E$ has two edges or more, then $\Gamma' > (n+1) \Lambda = \Gamma$. This implies that $m_V(\ubar_1), m_V(\ubar_2), \ldots, m_V(\ubar_n), m_V(\ubar_1)$ is a Hamiltonian cycle of $\Clique$. This cycle is of length $\Gamma' - n \Lambda \le (n+1) \Lambda - n \Lambda = \Lambda $.
    
    We can similarly show that C-$\langle \Path_r \rightarrow \Clique_s \rangle$ is \np-complete.
    %Using an analogous reduction from the \textsc{Path-Traveling Salesman Problem} \cite{garey1974problems}, we can show that C-$\langle \Path_r \rightarrow \Clique_s \rangle$ is \np-complete. \todo{citation!}
\end{proof}

\subsection{Virtual wheel and tree}

The Proposition~\ref{proposition:section3:cyclepath:hamilto} cannot be extended to the virtual wheel and tree cases, for which there does not necessarily exists a feasible mapping when $n_r > n_s$.

\vspace{0.25cm}
\begin{open}\label{theorem:section3:wheeltree:open}
    Are E-$\langle \Wheel_r \rightarrow \Clique_s \rangle$ and  E-$\langle \Tree_r \rightarrow \Clique_s \rangle$ \np-complete ?
\end{open}
\vspace{0.25cm}

For the cost variant, C-$\langle \Tree_r \rightarrow \Clique_s \rangle$ is \np-complete, even if $n_r \le n_s$, as a consequence of Theorem~\ref{theorem:section3:cyclepath:cost}.  

\vspace{0.25cm}
\begin{theorem}\label{theorem:section3:wheeltree:cost}
    C-$\langle \Wheel_r \rightarrow \Clique_s \rangle$ is \np-complete
\end{theorem}

\begin{proof}
    Consider a \tsp instance $(\Clique, d, \Lambda)$. An instance of C-\univne can be constructed as follows.
    The virtual graph is a wheel $\Wheel_r$ of $n+1$ nodes.
    The substrate graph $\Clique_s$ is $\Clique$ with an additional node $u_{n+1}$ connected to every other node, hence it is a clique with $n+1$ nodes. It has unit capacity on every node and edge. Its costs are: $w(u_i) = 0$, for $i \in \{1, \ldots, n\}$, $w(u_i, u_j) = \Lambda + d(u_i, u_j)$ for $i \neq j \in \{1, \ldots, n\}$ (called \textit{expensive edges}), and $w(u_i, u_{n+1}) = 0$ for $i \in \{1, \ldots, n\}$ (called \textit{cheap edges}). 
    The bound is $\Gamma = (n+1) \Lambda$.

    We show that $\Clique$ contains a Hamiltonian cycle of length at most $\Lambda$ if and only if there is a feasible mapping of $\Wheel_r$ on $\Clique_s$ of cost at most $\Gamma'$.

    ($\Rightarrow$) Consider $\Clique$ contains a Hamiltonian cycle $u_{\pi(1)}, u_{\pi(2)}, \ldots, u_{\pi(n)}$ of length $\Lambda' \le \Lambda$. The following mapping $m = (m_V, m_E)$ is clearly feasible, and $W(m) = n \Lambda + \Lambda' \le (n+1)\Lambda = \Gamma$:
    
    \noindent $m_V(\ubar_i) = u_{\pi(i)}$ and $m_E(\ubar_i, \ubar_{i+1}) = \{(u_{\pi(i)}, u_{\pi(i+1)})\}$, for $i \in \{1, \ldots, n\}$ \\
    $m_V(\ubar_c) = u_{n+1}$ and $m_E(\ubar_i, \ubar_c) = \{(u_{\pi(i)}, u_{n+1})\}$, for $i \in \{1, \ldots, n\}$.
    
    ($\Leftarrow$) Suppose that there is a feasible mapping $m = (m_V, m_E)$ of cost $\Gamma' \le \Gamma$.
    It can be shown, as done in the previous proof, that $n$ virtual edges are routed on a single expensive edge, the remaining $n$ virtual edges being routed on a single cheap edge.
    Note that, only $u_{n+1}$ can host $\ubar_c$. Indeed, suppose otherwise, i.e. $u_{n+1}$ hosts $\ubar_i$, with $i \in \{1, \ldots, n\}$. Then, only $(\ubar_{i-1}, \ubar_i)$, $(\ubar_{i-1}, \ubar_i)$ and $(\ubar_{i-1}, \ubar_i)$ are routed on a single cheap edge, which contradicts that $n$ edges must be routed on a single cheap edge. Thus $u_{n+1}$ hosts $\ubar_c$, and it follows that the routing of the virtual edges $(\ubar_i, \ubar_{n+1})$, $i \in \{1, \ldots, n\}$, uses the $n$ cheap edges. Thus each edge of the outer cycle of the wheel is routed on a single expensive edge. This implies that $m_V(\ubar_1), m_V(\ubar_2), \ldots, m_V(\ubar_n)$ is a Hamiltonian cycle of $\Clique$ of length $\Gamma' - n \Lambda \le (n+1) \Lambda = \Lambda$.
\end{proof}

\subsection{Virtual clique}

The following result, which is shown extending the proof of Theorem~\ref{theorem:section2:cyclewheel:E-KrGs}, implies that E-$\langle \Graph_r \rightarrow K_s \rangle$ is \np-complete when $n_r > n_s$:

\vspace{0.25cm}
\begin{theorem}\label{theorem:section3:clique:exist}
    E-$\langle  \Clique_r \rightarrow \Clique_s \rangle$ is $\mathcal{NP}$-hard.
\end{theorem}

\begin{proof}    
    Consider an \edpp instance ($\Graph, (s_l, t_l)_{l \in \{1, \ldots, k\}}$), with nodes of $\Graph$ denoted as $u_1, \ldots, u_n$. For each terminal node $s_i$ (resp. $t_i$), $i\in \{1, \ldots, k\}$, there exists a node $u_j \in $ $j \in \{1, \ldots, n\}$, such that $s_i = u_j$ (resp. $t_i = u_j)$.  In the following, both notations might be used to represent the same node. 
    An instance of E-\univne can be constructed as follows.     
    The virtual graph is a clique $\Clique_r$ with $n + 2k$ nodes $\ubar_1, \ldots, \ubar_n, \sbar_1, \tbar_1, \ldots, \sbar_k, \tbar_k$.
    The substrate graph is a clique $\Clique_s$ with the nodes of $\Graph$. A substrate node has a capacity of two for terminal nodes, and one otherwise. Substrate edges have a capacity of one, plus one if the edge is already in $\Graph$, plus one if the edge links a pair of terminal nodes not in the same pair.

    Let us now show that there are edge disjoint paths in $\Graph$ if and only if there is a feasible mapping of $\Clique_r$ on $\Graph_s$.

    ($\Rightarrow$) Suppose there are $k$ disjoint paths $P_1, \ldots, P_k$ between the pairs of terminal nodes. The following mapping $m = (m_V, m_E)$ can be verified to be valid and feasible: \\
    \indent $m_V(\ubar_i) = u_i$, $m_E(\ubar_i, \ubar_j) = \{(u_i, u_j)\}$, for $i \neq j \in \{1, \ldots, n\}$ \\
    \indent $m_V(\sbar_l) = s_l$,~ $m_V(\tbar_l) = t_l$, $m_E(\sbar_l, \tbar_l) = P_l$, ~ for $l \in \{1, \ldots, k\}$~ and \\
    \indent $m_E(\sbar_{l}, \sbar_{l'}) = \{(s_{l}, s_{l'})\}$, ~$m_E(\sbar_{l}, \tbar_{l'}) = \{(s_{l}, t_{l'})\}$,~ $m_E(\tbar_{l}, \tbar_{l'}) = \{(t_{l}, t_{l'})\}$, ~ for $l \neq l' \in \{1,\ldots,  k\}$, \\
    \indent $m_E(\ubar_i, \sbar_l) = \{(u_i, s_l)\}$, ~$m_E(\ubar_i, \tbar_l) = \{(u_i, t_l)\}$,~ for $i \in \{1, \ldots, n\}$ and $l \in \{1,\ldots,  k\}$.

    ($\Leftarrow$) Suppose that there is a feasible mapping $m = (m_V, m_E)$. Since every node is interchangeable in a uniform-demand virtual clique, we assume w.l.o.g. that $m_V(\ubar_i) = u_i$, for $i \in \{1, \ldots, n\}$,  and that $m_V(\sbar_l) = s_l$ and $m_V(\tbar_l) = t_l$, for $l \in \{1, \ldots, k\}$. \\
    \noindent First consider that $(\ubar_{i}, \ubar_{j})$, $(\ubar_{i}, \sbar_{l})$, $(\ubar_{i}, \tbar_{l})$, $(\sbar_{l}, \sbar_{l'})$, $(\sbar_{l}, \tbar_{l'})$ and $(\tbar_{l}, \tbar_{l'})$, for $i \neq j \in \{1, \ldots, n\}$ and $l \neq l' \in \{1, \ldots, k\}$, are \textit{correctly} routed, that is: $m_E(\ubar_{i}, \ubar_{j}) = (u_{i}, u_{j})$, $m_E(\ubar_{i}, \sbar_{l}) = (u_{i}, s_{l})$, $m_E(\ubar_{i}, \tbar_{l}) = (u_{i}, t_{l})$, $m_E(\sbar_{l}, \sbar_{l'}) = (s_{l}, s_{l'})$, $m_E(\sbar_{l}, \tbar_{l'}) = (s_{l}, t_{l'})$, $m_E(\tbar_{l}, \sbar_{l'}) = (t_{l}, t_{l'})$. Note that such partial routing lets exactly one unit of capacity on edges that exists in $\Graph$ (and none on additional edges). The paths $m_E(\sbar_l, \tbar_l)$, $l \in  \{1, \ldots, k\}$, are then edge disjoint on $\Graph$, and are thus a solution of \edpp.
    \noindent On the contrary, let us consider there is one \textit{incorrectly} routed virtual edge. Then, as done in proof of Theorem~\ref{theorem:section2:cyclewheel:E-KrGs}, it can be shown that a new mapping, that routes correctly all edges, can be constructed. 
\end{proof}

This result implies the \np-completeness of the cost variant, in the general case. However, it does not when $n_r \le n_s$. For this case, we have not been able to adapt the reduction of \tsp.

\vspace{0.25cm}
\begin{open}
    When $n_r \le n_s$, is C-$\langle \Clique_r \rightarrow \Clique_s \rangle$ \np-complete ?
\end{open}

\section{Substrate tree}

A dynamic programming algorithm with a $\BigO(3^{n_r} n_s)$ computation time is proposed when the substrate graph is a directed tree by \cite{figiel2021tree}. Using the decomposition scheme of this exponential algorithm, we propose polynomial time dynamic programming algorithms for embedding undirected virtual paths, cycles, wheels and cliques on an undirected substrate tree. However, we show that the cases of a general virtual graph and of a virtual tree on a substrate tree remain \np-complete.

Given a substrate tree $\Tree_s = (V_s, E_s)$, we consider that $\Tree_s$ is rooted, i.e. a node $u_r \in V_s$ is its {\it root}. 
Given $u \in V_s$ and $v \in V_s$, the unique path between $u$ and $v$ is denoted $p(u, v)$.
If $v$ is on the path $p(u_r, u)$, $u$ is a descendant of $v$, moreover, if $(u,v) \in E_s$, $v$ (resp. $u$) is said to be the parent of $u$ (resp. the children of $v$). A leaf is a node that has no children. 
The subtree of $\Tree_s$ induced by $u$ and its descendants is denoted $\Tree_s(u)$.

\subsection{Virtual cycle} \label{section4:cycle}

In section~\ref{section4:cycle} and section~\ref{section4:rest}, the tree is assumed to be binary and with a null capacity on non-leaf nodes. Indeed, it is shown by \cite{figiel2021tree} that a \vne instance with an arbitrary substrate tree can be reduced to such a case in polynomial time by adding some dummy nodes.

Given a virtual cycle $\Cycle_r$, the $k$ nodes $\ubar_i, \ubar_{i+1},\ldots, \ubar_{i+k-1}$, $i \in \{1,\ldots, n_r\}$, are said to be consecutive, where $\ubar_i$ (resp. $\ubar_{i+k-1})$ is called the first (resp. last) node. We first give a technical Lemma about the specific structure of optimal solutions.

\vspace{0.25cm}
\begin{lemma}\label{lemma:section4:cycle:consec}
    There is a minimum cost feasible mapping of $\Cycle_r$ on $\Tree_s$ where any subtree hosts consecutive nodes of $\Cycle_r$.
\end{lemma}

\begin{proof}
    Let $m = (m_V, m_E)$ be a feasible mapping  of $\Cycle_r$ on $\Tree_s$ where the nodes hosted on some subtrees of $\Tree_s$ are not consecutive. We will construct a feasible mapping $m^*$ from $m$, of lower or equal cost such that the property holds. 
    This construction is done through an iterative process that starts on the root of the tree, and then progress toward the leaves. On each node, a number of iterations happens as follows. 
    
    Consider a node $u \in V_s$. Let $u_a$ and $u_b$ be the children of $u$, and $A$ (resp. $B$) be the nodes placed on $\Tree_s(u_a)$ (resp. $\Tree_s(u_b)$), such that $A$ and $B$ are not consecutive nodes. Let us consider their partitions $A = A_1 \cup A_2 \cup \ldots \cup A_{n_A}$ and $B = B_1 \cup B_2 \cup \ldots \cup B_{n_B}$ whose elements are consecutive nodes. One iteration of the process, that we detail below, gives a feasible mapping $m'$, with a lower or equal cost than $m$, where the nodes $A'$ (resp. $B'$) placed on $\Tree_s(u_a)$ (resp. $\Tree_s(u_b)$) are such that their partition into subsets of consecutive nodes both contain one element less, compared to $A$ and $B$. After some iterations, a mapping such that the nodes placed on $\Tree_s(u_a)$ and $\Tree_s(u_b)$ are consecutive is obtained. Then the process move on the following node, e.g. $u_a$ or $u_b$.

    The key idea is to swap the node placement of $A_1$ and $B_2$ as follows.
    \begin{align*}
        & m'_V(\ubar_i)  = m_V(\ubar_i), && \forall i \in \{1, \ldots, |A_1|\} \\
        & m'_V(\ubar_i) = m_V(\ubar_{i+|B_1|}), && \forall i \in \{|A_1| + 1, \ldots, |A_1| + |A_2| \} \\
        & m'_V(\ubar_i) = m_V(\ubar_{i-|A_2|}), && \forall i \in \{|A_1| + |A_2| + 1, \ldots, |A_1| + |A_2| + |B_1|\} \\
        & m'_V(\ubar_i) = m_V(\ubar_i), && \forall i \in \{|A_1| + |A_2| + |B_1| + 1, \ldots, |A| + |B| \}
    \end{align*}

    Figure~\ref{fig:section4:cycle:swap} a) provides an example of a mapping for which the property does not hold. The virtual graph is the cycle shown on the left, and the substrate graph is the tree $\Tree_s$ shown on the right. The mapping is the following: a virtual node is shown below the substrate node it is placed on, and the routing of the edges is shown on $\Tree_s$. The virtual nodes colored in green (resp. blue) are placed on $\Tree_s(u_a)$ (resp. $\Tree_s(u_b)$).
    Figure~\ref{fig:section4:cycle:swap} b) provides analogously the mapping obtained after one iteration of the process. Note that, in this case, the overall algorithm is completed after that iteration.
    
    \begin{figure}
    \centering
    \begin{tikzpicture}
        % first mapping
        \begin{scope}[yshift = 2.75cm]
            
            \node at (-5.25, 0){$a)$};

            % virtual cycle
            \begin{scope}[xshift = -3.1cm]
                \node[vnode, fill=white!40!green, minimum size=0.4cm] (v1) at (2*360/7:1.4){$\ubar_1$};
                \node[vnode, fill=white!40!green, minimum size=0.4cm] (v2) at (1*360/7:1.4){$\ubar_2$};
                \node[vnode, fill=white!40!blue, minimum size=0.4cm] (v3) at (7*360/7:1.4){$\ubar_3$};
                \node[vnode, fill=white!40!blue, minimum size=0.4cm] (v4) at (6*360/7:1.4){$\ubar_4$};
                \node[vnode, fill=white!40!green, minimum size=0.4cm] (v5) at (5*360/7:1.4){$\ubar_5$};
                \node[vnode, fill=white!40!blue, minimum size=0.4cm] (v6) at (4*360/7:1.4){$\ubar_6$};
                \node[vnode, fill=white!40!blue, minimum size=0.4cm] (v7) at (3*360/7:1.4){$\ubar_7$};
                \draw[networkedge] (v1) -- (v2);
                \draw[networkedge] (v2) -- (v3);
                \draw[networkedge] (v3) -- (v4);
                \draw[networkedge] (v4) -- (v5);
                \draw[networkedge] (v5) -- (v6);
                \draw[networkedge] (v6) -- (v7);
                 \draw[networkedge, red] (v7) to (v1);
                \draw[networkedge, teal] (v2) to (v3);
                \draw[networkedge, yellow] (v4) to (v5);
                \draw[networkedge, pink] (v5) to (v6);
            \end{scope}
        
            % substrate tree
            \begin{scope}[xshift = 3.1cm]
                \node[snode] (ur) at (0,2){$u_r$};
                \node[snode] (ua) at (-2.,0.5){$u_a$};
                \node[snode] (ub) at (2.,0.5){$u_b$};
                \node[snode] (u1) at (-3.,-1){$u_1$};
                \node[snode] (u2) at (-1.,-1){$u_2$};
                \node[snode] (u3) at (1.,-1){$u_3$};
                \node[snode] (u4) at (3.,-1){$u_4$};
                \draw[networkedge] (ur) to (ua);
                \draw[networkedge] (ur) to (ub);
                \draw[networkedge] (ua) to (u1);
                \draw[networkedge] (ua) to (u2);
                \draw[networkedge] (ub) to (u3);
                \draw[networkedge] (ub) to (u4);
    
                % mapping
                \node[vnode, fill=white!40!green, minimum size=0.4cm] (v1) at (-3., -1.75){$\ubar_1$};
                \node[vnode, fill=white!40!green, minimum size=0.4cm] (v2) at (-3., -2.5){$\ubar_2$};
                \node[vnode, fill=white!40!green, minimum size=0.4cm] (v3) at (-1., -1.75){$\ubar_5$};
                \node[vnode, fill=white!40!blue, minimum size=0.4cm] (v4) at (1., -1.75){$\ubar_3$};
                \node[vnode, fill=white!40!blue, minimum size=0.4cm] (v5) at (1., -2.5){$\ubar_4$};
                \node[vnode, fill=white!40!blue, minimum size=0.4cm] (v6) at (3., -1.75){$\ubar_6$};
                \node[vnode, fill=white!40!blue, minimum size=0.4cm] (v7) at (3., -2.5){$\ubar_7$};
                
                \draw[networkedge] (v1) to (v2);
                \draw[routingedge] (v4) to (v5);
                \draw[routingedge] (v6) to (v7);
                
                \draw[routingedge, teal] (v2) to[bend right = 40] (u1);
                \draw[routingedge, teal] (u1) to[bend right = 12] (ua);
                \draw[routingedge, teal] (ua) to[bend right = 12] (ur);
                \draw[routingedge, teal] (ur) to[bend right = 12] (ub);
                \draw[routingedge, teal] (ub) to[bend left = 12] (u3);
                \draw[routingedge, teal] (u3) to (v4);
            
                \draw[routingedge, yellow] (v5) to[bend right = 40] (u3);
                \draw[routingedge, yellow] (u3) to[bend right = 24] (ub);
                \draw[routingedge, yellow] (ub) to[bend left = 24] (ur);
                \draw[routingedge, yellow] (ur) to[bend left = 24] (ua);
                \draw[routingedge, yellow] (ua) to[bend right = 12] (u2);
                \draw[routingedge, yellow] (u2) to[bend right = 12] (v3);
    
                \draw[routingedge, pink] (v3) to[bend right = 12] (u2);
                \draw[routingedge, pink] (u2) to[bend left = 24] (ua);
                \draw[routingedge, pink] (ua) to[bend right = 36] (ur);
                \draw[routingedge, pink] (ur) to[bend right = 36] (ub);
                \draw[routingedge, pink] (ub) to[bend right = 12] (u4);
                \draw[routingedge, pink] (u4) to[bend left = 15] (v6);
    
                \draw[routingedge, red] (v7) to[bend right = 40] (u4);
                \draw[routingedge, red] (u4) to[bend right = 12] (ub);
                \draw[routingedge, red] (ub) to[bend right = 12] (ur);
                \draw[routingedge, red] (ur) to[bend right = 12] (ua);
                \draw[routingedge, red] (ua) to[bend right = 12] (u1);
                \draw[routingedge, red] (u1) to[bend right = 10] (v1);
            \end{scope}

        \end{scope}
    
        % second mapping
        \begin{scope}[yshift = -2.75cm]
        
            \node at (-5.25, 0){$b)$};
            
            % virtual cycle
            \begin{scope}[xshift = -3.1cm]
                \node[vnode, fill=white!40!green, minimum size=0.4cm] (v1) at (2*360/7:1.4){$\ubar_1$};
                \node[vnode, fill=white!40!green, minimum size=0.4cm] (v2) at (1*360/7:1.4){$\ubar_2$};
                \node[vnode, fill=white!40!green, minimum size=0.4cm] (v3) at (7*360/7:1.4){$\ubar_3$};
                \node[vnode, fill=white!40!blue, minimum size=0.4cm] (v4) at (6*360/7:1.4){$\ubar_4$};
                \node[vnode, fill=white!40!blue, minimum size=0.4cm] (v5) at (5*360/7:1.4){$\ubar_5$};
                \node[vnode, fill=white!40!blue, minimum size=0.4cm] (v6) at (4*360/7:1.4){$\ubar_6$};
                \node[vnode, fill=white!40!blue, minimum size=0.4cm] (v7) at (3*360/7:1.4){$\ubar_7$};
                \draw[networkedge] (v1) -- (v2);
                \draw[networkedge] (v2) -- (v3);
                \draw[networkedge] (v3) -- (v4);
                \draw[networkedge] (v4) -- (v5);
                \draw[networkedge] (v5) -- (v6);
                \draw[networkedge] (v6) -- (v7);
                \draw[networkedge, red] (v7) to (v1);                
                \draw[networkedge, teal] (v2) to (v3);
                \draw[networkedge, yellow] (v3) to (v4);
                \draw[networkedge, pink] (v5) to (v6);
            \end{scope}
        
            % substrate tree
            \begin{scope}[xshift =3.1cm]
                \node[snode] (ur) at (0,2){$u_r$};
                \node[snode] (ua) at (-2.,0.5){$u_a$};
                \node[snode] (ub) at (2.,0.5){$u_b$};
                \node[snode] (u1) at (-3.,-1){$u_1$};
                \node[snode] (u2) at (-1.,-1){$u_2$};
                \node[snode] (u3) at (1.,-1){$u_3$};
                \node[snode] (u4) at (3.,-1){$u_4$};
                \draw[networkedge] (ur) to (ua);
                \draw[networkedge] (ur) to (ub);
                \draw[networkedge] (ua) to (u1);
                \draw[networkedge] (ua) to (u2);
                \draw[networkedge] (ub) to (u3);
                \draw[networkedge] (ub) to (u4);

                % mapping
                \node[vnode, fill=white!40!green, minimum size=0.4cm] (v1) at (-3., -1.75){$\ubar_1$};
                \node[vnode, fill=white!40!green, minimum size=0.4cm] (v2) at (-3., -2.5){$\ubar_2$};
                \node[vnode, fill=white!40!green, minimum size=0.4cm] (v3) at (-1., -1.75){$\ubar_3$};
                \node[vnode, fill=white!40!blue, minimum size=0.4cm] (v4) at (1., -1.75){$\ubar_4$};
                \node[vnode, fill=white!40!blue, minimum size=0.4cm] (v5) at (1., -2.5){$\ubar_5$};
                \node[vnode, fill=white!40!blue, minimum size=0.4cm] (v6) at (3., -1.75){$\ubar_6$};
                \node[vnode, fill=white!40!blue, minimum size=0.4cm] (v7) at (3., -2.5){$\ubar_7$};
                
                \draw[networkedge] (v1) to (v2);
                \draw[networkedge] (v4) to (v5);
                \draw[networkedge] (v6) to (v7);
                
                \draw[networkedge, teal] (v2) to[bend right = 40] (u1);
                \draw[networkedge, teal] (u1) to[bend right = 12] (ua);
                \draw[networkedge, teal] (ua) to[bend right = 12] (u2);
                \draw[networkedge, teal] (u2) to (v3);

                \draw[networkedge, yellow] (v3) to[bend right = 40] (u2);
                \draw[networkedge, yellow] (u2) to[bend left = 24] (ua);
                \draw[networkedge, yellow] (ua) to[bend right = 12] (ur);
                \draw[networkedge, yellow] (ur) to[bend right = 12] (ub);
                \draw[networkedge, yellow] (ub) to[bend left = 12] (u3);
                \draw[networkedge, yellow] (u3) to[bend right = 12] (v4);
    
                \draw[networkedge, pink] (v5) to[bend right = 40] (u3);
                \draw[networkedge, pink] (u3) to[bend right = 24] (ub);
                \draw[networkedge, pink] (ub) to[bend right = 12] (u4);
                \draw[networkedge, pink] (u4) to[bend right = 15] (v6);
    
                \draw[networkedge, red] (v7) to[bend right = 40] (u4);
                \draw[networkedge, red] (u4) to[bend right = 12] (ub);
                \draw[networkedge, red] (ub) to[bend right = 12] (ur);
                \draw[networkedge, red] (ur) to[bend right = 12] (ua);
                \draw[networkedge, red] (ua) to[bend right = 12] (u1);
                \draw[networkedge, red] (u1) to[bend right = 10] (v1);
            
            \end{scope}

        \end{scope}

    \end{tikzpicture}
    \caption{Example of a node placement swap from $a)$ to $b)$ from proof of Lemma~\ref{lemma:section4:cycle:consec}}
    \label{fig:section4:cycle:swap}
    \end{figure}
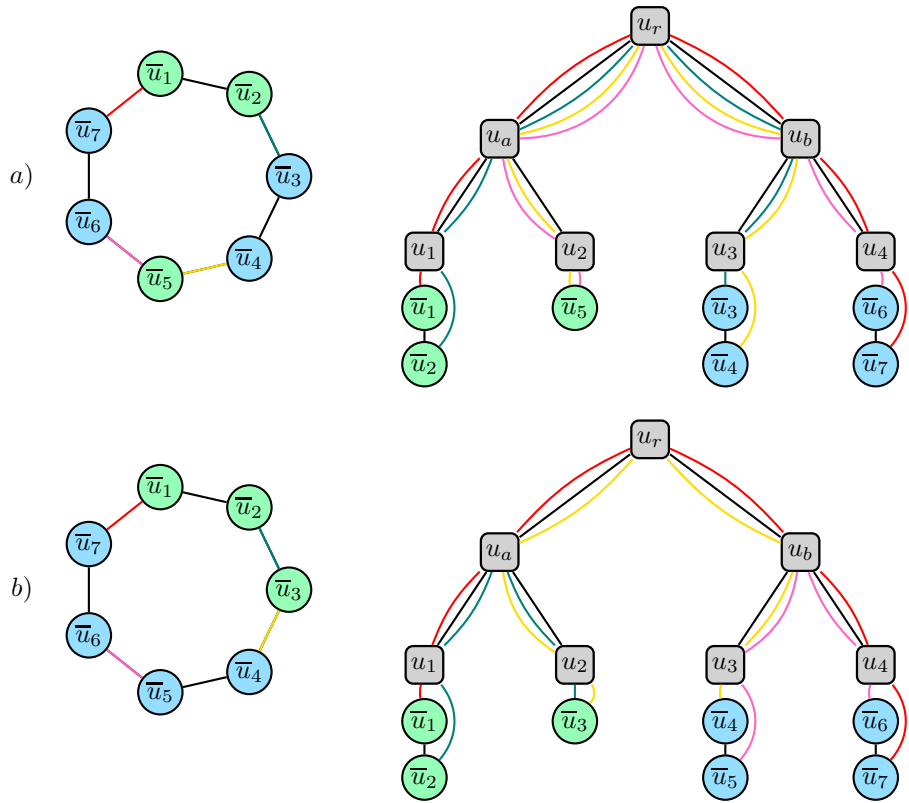

    In the new mapping $m'$, every substrate node will host the same number of virtual nodes as in $m$, thus the node placement is feasible and has the same cost.

    The edge routing $m'_E$ is directly deduced from $m'_V$ as there is only one path between two substrate nodes. It can be described completely from the paths of $m_E$, as follows.
    
    Let $x$ denote the last node to appear in both paths $p(u, m_V(\ubar_{|A_1|}))$ and $p(u, m_V(\ubar_{|A_1| + |B_1| + 1}))$; $y$ the last node to appear in both paths $p(u, m_V(\ubar_{|A_1|}))$ and $p(u, m_V(\ubar_{|A_1| + |B_1| + |A_2|}))$; $z$ the last node to appear in both paths $p(u, m_V(\ubar_{|A_1| + |B_1|}))$ and $p(u, m_V(\ubar_{|A_1| + |B_1| + |A_2| + 1}))$. 
    In Figure~\ref{fig:section4:cycle:swap}~a), $x = y = u_a$ and $z = u_b$.
    Note that $x, y \in T_s(u_a)$ and $z \in T_s(u_b)$. We will consider in the following the case where $y$ in on the path $p(x, u)$, with the proof in the other case being analogous. 

    These nodes appear notably in the original routing paths of the three virtual edges:
    \begin{align*}
        & m_E(\ubar_{|A_1|}, \ubar_{|A_1|+1}) && =  && p(m_V(\ubar_{|A_1|}), x) \cup p(x, y) \cup p(y, m_V(\ubar_{|A_1|+1})) \\
        & m_E(\ubar_{|A_1| + |B_1|}, \ubar_{|A_1| + |B_1| +1}) && = &&p(m_V(\ubar_{|A_1| + |B_1|}, z) \cup p(z, x) \cup p(x, m_V(\ubar_{|A_1| + |B_1|+1})) \\
        & m_E(\ubar_{|A_1| + |B_1| + |A_2|}, \ubar_{|A_1| + |B_1| + |A_2| +1}) && = &&  p(m_V(\ubar_{|A_1| + |B_1| + |A_2|}, y) \cup p(y, z) \\
        & && &&\cup  p(y, m_V(\ubar_{|A_1| + |B_1| + |A_2|+1}))
    \end{align*}

    The routing, in $m'_E$, of edges $(\ubar_{|A_1|}, \ubar_{|A_1|+1})$, $( \ubar_{|A_1|+|A_2|},\ubar_{|A_1|+|A_2|+1})$, and $(\ubar_{|A_1|+|A_2|+|B_1|}, \ubar_{|A_1|+|A_2|+|B_1|+1})$ uses combinations of the subpaths defined above as follows:
    \begin{align*}
        & m'_E(\ubar_{|A_1|}, \ubar_{|A_1|+1}) && = &&p(m_V(\ubar_{|A_1|}), x) \cup p(x, m_V(\ubar_{|A_1| + |B_1|+1})) \\
        & m'_E(\ubar_{|A_1| + |A_2|}, \ubar_{|A_1| + |A_2| + 1}) && = && p(m_V(\ubar_{|A_1| + |B_1| + |A_2|}, y) \cup p(y, m_V(\ubar_{|A_1|+1})) \\
        & m'_E(\ubar_{|A_1| + |A_2| + |B_1|}, \ubar_{|A_1| + |A_2| + |B_1| + 1}) && = && p(m_V(\ubar_{|A_1| + |B_1|}, z) \cup p(z, m_V(\ubar_{|A_1| + |A_2| + |B_1|})
    \end{align*}
 
    Remark that, in these paths, the three subpaths $p(x,y)$, $p(z, x)$ and $p(y,z)$ do not appear.
    \noindent For the rest of the virtual edges, their routings in $m'_E$ use exactly the paths of $m_E$, as follows:
    \begin{align*}
        & m'_E(\ubar_{i}, \ubar_{i+1}) = m_E(\ubar_{i}, \ubar_{i+1}), &\quad& \forall i \in \{1, \ldots, |A_1| -1\} \\
        & m'_E(\ubar_{i}, \ubar_{i+1}) = m_E(\ubar_{i+|B_1|}, \ubar_{i+1+|B_1|}), && \forall i \in \{ |A_1|+ 1, \ldots, |A_1| + |A_2| -1\} \\
        & m'_E(\ubar_{i}, \ubar_{i+1}) = m_E(\ubar_{i-|A_2|}, \ubar_{i+1-|A_2|}), && \forall i \in \{|A_1| + |A_2| + 1, \ldots, |A_1| + |A_2| + |B_1| -1\} \\
        & m'_E(\ubar_{i}, \ubar_{i+1}) = m_E(\ubar_{i}, \ubar_{i+1}), && \forall i \in \{|A_1| + |A_2| + |B_1| + 1, \ldots, n_r\}
    \end{align*}

    In the new mapping $m'$, every substrate edge will host the same number of virtual edges as in $m$, except the substrate edges of $p(x,y)$, $p(z,x)$ and $p(y,z)$ which will all host fewer virtual edges. Note that these paths are not empty, since $x, y \in T_s(u_a)$ and $z \in T_s(u_b)$. 
    Thus the edge routing is feasible and has a lower cost: $W(m'_E) = W(m_E) - \sum_{e \in p(x, y)} w(e) - \sum_{e \in p(z, x)} w(e) - \sum_{e \in p(y, z)} w(e) \le W(m_E)$.
\end{proof}

We now provide a dynamic algorithm to solve the problem. From Lemma~\ref{lemma:section4:cycle:consec}, only mappings that place consecutive nodes on every subtrees require to be considered. The key idea of the algorithm is to relax the knowledge of which consecutive nodes of $\Cycle_r$ are placed on a subtree, considering instead only their number. Let us now introduce some notations.

Let $W_E(u, v, k)$ denotes the part of a mapping cost induced by the utilization for a substrate edge $(u,v) \in E_s$ such that $v$ is the parent of $u$, when $k$ consecutive nodes are placed on $\Tree_s(u)$. If $ 0 < k < n_r$, $(u,v)$ hosts exactly two edges, thus $W_E(u, v, k) = 2 * w(u,v)$ (if $(u,v)$ has a unit capacity, such a mapping would be unfeasible and the cost is set to $+\infty$). Otherwise, if  $k = n_r$ (resp. $k=0$) then the cycle is entirely placed on (resp. outside of) $\Tree_s(u)$, and $(u,v)$ hosts no virtual edge.

Let $W^*(u, k)$ denote the minimum mapping costs induced by the utilization of subtree $\Tree_s(u)$, when $k$ consecutive nodes are placed on it. It includes the (placement and routing) costs of using all the substrate (node and edge) resources in the subtree. 
Note that this mapping cost is the same for any $k$ consecutive nodes considered. Indeed, given a mapping that places nodes $\ubar_i, \ubar_{i+1}, \ldots, \ubar_{i+k-1}$, $i \in \{1, \ldots, n_r\}$, on $\Tree_s(u)$, then a mapping of same restricted cost, and overall cost, that places $\ubar_j, \ubar_{j+1}, \ldots, \ubar_{j+k-1}$, $j \in \{1, \ldots, n_r\}$, on $\Tree_s(u)$ can be obtained simply by shifting the indices of the virtual node in the placement by $j-i$.
This minimum restricted cost is denoted $W^*(u, k)$.
If the capacities are insufficient, the cost is again set to $+\infty$.

If $u$ is a leaf node, clearly $W^*(u, k) = k * w(u)$ (if $k \le c(u)$). Otherwise, $u$ has a children $u_a$ and $u_b$, the following recurrence formula holds:

\begin{equation} \label{equat:section4:cycle:recur}
    W^*(u, k) = \min_{l_a + l_b = k} W^*(u_a, l_a) + W_E(u_a, u, l_a) + W^*(u_b, l_b) + W_E(u_b, u, l_b) 
\end{equation}

\noindent Indeed, consider the mapping $m = (m_V, m_E)$ that has minimum restricted costs $W^*(u, k)$ on $\Tree_s(u)$. Let us denote $l_a$ (resp. $l_b$) the number of virtual nodes placed on $\Tree(u_a)$ (resp. $\Tree_s(u_b)$). Recall that only leaf nodes have capacity, thus $u$ host no nodes and $l_b = k - l_a$. We can consider, from Lemma~\ref{lemma:section4:cycle:consec}, that both these set of nodes are consecutive. Moreover, the mapping of the $l_a$ (resp. $l_b$) consecutive nodes placed on $\Tree_s(u_a)$ (resp. $\Tree_s(u_b)$) must have minimum restricted costs on $\Tree_s(u_a)$ (resp. $\Tree_s(u_b)$). Otherwise, there is another mapping $m' = (m'_V, m'_E)$ with better restricted costs for $l_a$ nodes on $\Tree_s(u_a)$. Then, we could construct another mapping that is $m$, except for the placement on $\Tree_s(u_a)$, that is from $m'$. That mapping would then have a better restricted cost over $\Tree_s(u)$, which is a contradiction.

\vspace{0.25cm}
\begin{theorem}\label{theorem:section4:cycle:cost}
    C-$\langle \Cycle_r \rightarrow \Tree_s \rangle$ is in \p and is solvable in $\BigO({n_r}^2 n_s)$
\end{theorem} 

\begin{proof}
    Consider the following leaf-to-root algorithm. First, the values of $W^*$ are computed for each leaf node, which takes $\BigO(n_r n_s)$. Then, for a non-leaf node, the values of $W^*$ are computed in $\BigO({n_r}^2)$, using formula~\ref{equat:section4:cycle:recur}. When $W^*(u_r, n_r)$ is computed, a solution is found. From Lemma~\ref{lemma:section4:cycle:consec}, this solution is optimal. Since the root node can only be visited by the dynamic algorithm at the last iteration, the overall computational time is $\BigO({n_r}^2 n_s)$.
\end{proof}

\subsection{Virtual path, wheel and clique}  \label{section4:rest}

The dynamic algorithm of Theorem~\ref{theorem:section4:cycle:cost} can be adapted for other virtual topologies than a cycle.

\paragraph{Virtual path}

The property given in Lemma~\ref{lemma:section4:cycle:consec} holds for a virtual path and we can construct by dynamic programming a minimum cost feasible mapping so that each subtree contains consecutive nodes. 

However, a given mapping, the cost of using an substrate edge $(u, v)$, with $v$ the parent of $u$, depends on whether $\Tree_s(u)$ hosts an extremity of $\Path_r$, then $(u, v)$ hosts exactly one edge; or none, then $(u, v)$ hosts exactly two edges.

The functions $W_E$ and $W^*$ thus require an additional parameter that indicates how many extremities are hosted on the considered subtree. The recurrence formula~\ref{equat:section4:cycle:recur} can be slightly modified accordingly. 

\vspace{0.25cm}
\begin{theorem}\label{theorem:section4:rest:pathcost}    
     C-$\langle \Path_r \rightarrow \Tree_s \rangle$ is in \p, and is solvable in $\BigO({n_r}^2 n_s)$
\end{theorem}

\paragraph{Virtual wheel}

The property given in Lemma~\ref{lemma:section4:cycle:consec} holds for the outer cycle of a wheel, and we can construct by dynamic programming a minimum cost feasible mapping so that each subtree contains consecutive nodes of the outer cycle. 

However, in a given mapping where $k$ (consecutive) nodes of the outer cycle are placed on a subtree $\Tree_s(u)$, the cost of using a substrate edge $(u, v)$, with $v$ the parent of $u$, depends on whether $\Tree_s(u)$ hosts the center of $\Wheel_r$, then $(u, v)$ hosts exactly $2 + n_r - k$ edges; or not, then $(u, v)$ hosts exactly  $2 + k$ edges.

The functions $W_E$ and $W^*$ thus require an additional parameter that indicates if the center of the wheel is hosted on the considered subtree. The recurrence formula~\ref{equat:section4:cycle:recur} can be slightly modified accordingly. 

\vspace{0.25cm}
\begin{theorem}\label{theorem:section4:rest:wheelcost}    
     C-$\langle \Wheel_r \rightarrow \Tree_s \rangle$ is in \p, and is solvable in $\BigO({n_r}^2 n_s)$
\end{theorem}

\paragraph{Virtual clique}

In the case of a virtual clique $\Clique_r$, for a given mapping that places exactly $k$ nodes on a subtree $\Tree_s(u)$, the edge $(u,v)$ hosts exactly $k (n_r-k)$ virtual edges. With function $W_E$ modified accordingly, the recurrence formula~\ref{equat:section4:cycle:recur} holds.

\vspace{0.25cm}
\begin{theorem}\label{theorem:section4:rest:cliquecost}
    C-$\langle \Clique_r \rightarrow \Tree_s \rangle$ is in \p, and can be solved in  $\BigO({n_r}^2 n_s)$.
\end{theorem}

\subsection{Virtual tree}

Recall that C-$\langle \Star_r \rightarrow \Star_s \rangle$ \cite{rost2015stars} and C-$\langle \Path_r \rightarrow \Star_s \rangle$ (as a corollary of Theorem~\ref{theorem:section4:rest:pathcost}) are in \p. We will prove the opposite result for virtual trees. Consider the following problem.

\begin{decproblem}
    \problemtitle{\bpplong (\bpp)}
    \probleminput{A set of $n$ integers $A = a_1, \ldots, a_n$, two integers $K$ (the number of bins) and $B$ (the size of each bin)}
    \problemquestion{Is there a partition of $A$ into $K$ disjoint sets $A_1, A_2, \ldots, A_K$, such that, for each subset $A_i$, $i \in \{1, \ldots, K\}$, $\sum_{a_l \in A_i} a_l \le B$ ?}
\end{decproblem}

In what follows, we reduce \bpp, which is $\mathcal{NP}$-hard \cite{garey1979guide}, to E-\univne. 
This reduction use similar gadgets that those used in proofs done by \cite{pankratov2023tree}, where \bpp is also reduced to a \vne case with a virtual tree (more specifically a \textit{tied-star}), and a substrate tree. However, their work consider that edge demands are non-uniform, and that $n_r = n_s$. Thus it needs to be adapted.

\vspace{0.25cm}
\begin{theorem}\label{theorem:section4:tree:exist}
    E-$\langle \Tree_r \rightarrow \Star_s \rangle$ is \np-complete.
\end{theorem}

\begin{proof}
    Consider a \bpp instance $(A, K, B)$ with $K \ge 3$. We assume w.l.o.g. that $\sum_{a_l \in A} a_l = K * B$, as adding some elements of unit weight to $A$ in any instance of \bpp makes it true. 
    In this proof, we will use a gadget which is a star $\Star(k)$, $k\in \mathbb{N}$, of $k$ leafs.
    An instance of E-\univne can be constructed as follows.
    The virtual graph is a tree $\Tree_r$, that we construct starting from a single node, the root $\ubar_r$. 
    For each element $l \in \{1, \ldots, n\}$, the center of a star gadget $\Star^l_r = \Star(a_l)$ is connected to $\ubar_r$. 
    Furthermore, $K * B - |A|$ \textit{singleton nodes} are connected to $\ubar_r$.
    The substrate graph is the star $\Star_s = \Star(K)$. The center of $\Star_s$ is denoted $u_c$ and has a unit capacity. The leafs of $\Star_s$, denoted $u_1, u_2, \ldots, u_K$, have $2*B$ capacity, and the edges $(u_i, u_c)$, $i \in \{1, \ldots, K\}$, have a $B$.

    We now show that there exists a solution to the \bpp instance if and only if there exists a feasible mapping of $\Tree_r$ on $\Star_s$:

    $(\Rightarrow)$: Consider a solution $A_1, \ldots, A_K$ of the \bpp instance.
    Consider the following mapping $m = (m_V, m_E)$: $\ubar_r$ is placed on $u_c$; all the nodes of $\Star^l_r$, for $a_l \in A_i$, are placed on $u_i$, for $i \in \{1, \ldots, K\}$. Furthermore $B - |A_i|$ singleton nodes are placed on $u_i$, for $i \in \{1, \ldots, K\}$. The routing of a virtual edge is the unique path between the two substrate nodes hosting the two virtual nodes connected by the edge. 
    This mapping is feasible, as $u_c$ hosts exactly one node, each substrate leaf hosts exactly $2 B$ virtual nodes, and each substrate edge hosts $B$ virtual edges.

    $(\Leftarrow)$: Suppose there is a feasible mapping $m = (m_V, m_E)$ of $\Tree_r$ on $\Star_s$. \\
    In a feasible mapping, only $u_c$ can host $\ubar_r$. Indeed, a leaf $u_i$, $i \in \{1, \ldots, K\}$, can only host a virtual node with at most $3 B - 1$ children, as $c(u_i) = 2 B$ and $c(u_i, u_c) = B$. However $\ubar_r$ has $K B \ge 3 B$ children, thus it can only be placed on $u_c$. \\
    Furthermore, each leaf $u_i$, $i \in \{1, \ldots, K\}$, hosts exactly $B$ children of $\ubar_r$. Indeed, $\ubar_r$, which is placed on $u_c$, has $K * B$ children and $\sum_{i \in \{1, \ldots, K\}} c(u_i, u_c) = K B$. 
    This indicates that, if the center of $\Star^l_r$, for $l \in A$, is placed on $u_i$, the leafs of $\Star^l_r$ are also placed on $u_i$, since there remains no capacity on edge $(u_i, u_c)$. 
    There remains, on $u_i$, a capacity of $B$ for those leafs. Thus the partition of $A$ into $A_i$, where $A_i$ are the elements whose stars are placed on $u_i$, $i \in \{1, \ldots, K\}$, is a valid solution for the \bpp instance.
\end{proof}

This shows that, even in graphs with uniform demands, packing aspects remain important for \vne problems, because of nodes connectivity in the virtual graph. 

The previous proof can be modified to show the result in the case where $n_r \le n_s$ for a general substrate tree. Given an instance of \bpp, the instance of \vne constructed has the same virtual tree, but the substrate graph is a tree $\Tree_s$, constructed from a root node $u_r$. For each of the $K$ bins, the center of a star $\Star(B-1)$ is connected to $u_r$. All nodes have a unit capacity, and all edges have a capacity of $B$. The rest of the proof is analogous to the above one.

\vspace{0.25cm}
\begin{theorem} \label{theorem:section4:tree:exist2}
    E-$\langle \Tree_r \rightarrow \Tree_s \rangle$ is \np-complete, even when $n_r \le n_s$
\end{theorem}

However, the reduction technique does not work for a substrate star when $n_r \le n_s$. Indeed, the \bpp case where there are more bins than items, that would construct such \univne instance, has always a trivial valid solution.

\vspace{0.25cm}
\begin{open}
    Is E-$\langle \Tree_r \rightarrow \Star_s \rangle$ \np-complete when $n_r \le n_s$ ?
\end{open}

\subsection{General virtual graph on a substrate path}

For a substrate path, the polynomial results obtained above holds, but not the \np-completness proof. However a reduction can be done using the following problem, which is \np-complete for planar and series-parallel graphs \cite{monien1988mincut}.

\begin{decproblem}
    \problemtitle{\mclalong (\mcla)}
    \probleminput{ An undirected graph $\Graph = (V, E)$, a positive integer $K$}
    \problemquestion{Is there a linear ordering of $\Graph$, i.e. a one-to-one function $f: V \rightarrow \{1, 2, \ldots, |V|\}$ such that, for $i \in \{1, \ldots, |V| - 1\}$, $|\{(\ubar, \vbar) \in E: f(\ubar) \le i < f(\vbar)\}| \le K$ ?}
\end{decproblem}

\begin{theorem} \label{theorem:section4:gen:exist}
    E-$\langle \Graph_r \rightarrow \Path_s \rangle$ is \np-complete
\end{theorem}

\begin{proof}
    Consider a \mcla instance $(\Graph, K)$. An instance of E-\univne can be constructed as follows. The virtual graph is $\Graph$. The substrate graph $\Path_r$ is a path with $|V|$ nodes $u_1, \ldots, u_{|V|}$. Every node has a capacity of one, and every edge has a capacity of $K$. 
    
    Let us show that there is a valid linear ordering of $\Graph$ if and only if there is a feasible mapping of $\Graph$ on $\Path_s$.

    $(\Rightarrow)$: Suppose that there is a linear ordering $f: V \rightarrow \{1, 2, \ldots, |V|\}$ of $\Graph$. A mapping $m = (m_V, m_E)$ can be constructed as follows: \\
    $m_V(\ubar) = u_{f(\ubar)}$, for $\ubar \in V$ and $m_E(\ubar, \vbar) = p(u_{f(\ubar)}, u_{f(\vbar)})$, for $(\ubar, \vbar) \in E$.
    
    We now show that this mapping is feasible. Since every substrate node host exactly one virtual node, $m_V$ is feasible. Note that a substrate edge $\{u_i, u_{i+1}\}$, $i \in \{1, \ldots, |V|-1|\}$, hosts an edge $\{\ubar, \vbar\}$ if and only if $f(\ubar) \le i$ and $f(\vbar) \ge i+1$. Hence the edge $\{u_i, u_{i+1}\}$ hosts\\
    $|(\ubar, \vbar) \in E : \{u_i, u_{i+1}\} \in m_E(\ubar, \vbar)| = |\{(\ubar, \vbar) \in E: f(\ubar) \le i < f(\vbar)\}| \le K$ virtual edges. Thus $m$ is feasible.

    $(\Leftarrow)$: Suppose that there is a feasible mapping $m = (m_V, m_E)$ of $\Graph_r$ on $\Path_s$. Consider the following linear ordering $f: V \rightarrow \{1, 2, \ldots, |V|\}$ where $f(\ubar)$ is the index of the substrate node that hosts $\ubar$.
    Since the substrate nodes have unit capacities, $f$ is a one-to-one function. Moreover, with the same argument on edge capacity as above,  $f$ is a valid linear ordering.
    %, that, for $i \in \{1, \ldots, |V| - 1\}$: \\
    %$|\{(\ubar, \vbar) \in E: f(\ubar) \le i < f(\vbar)\}| = |(\ubar, \vbar) \in E : (u_i, u_{i+1} \in m_E(\ubar, \vbar)| \le K$. Thus $f$ is a valid linear ordering.
\end{proof}

It has been shown that when $\Graph$ is a tree, \mcla is in \p \cite{yannakakis1985mincut}. However, the algorithm does not extend straightforwardly to the analog \univne case, since a substrate node might host several virtual nodes. Thus the result for \univne is an open problem.

\vspace{0.25cm}
\begin{open}\label{theorem:section4:gen:open}
    Is E-$\langle \Tree_r \rightarrow \Path_s \rangle$ \np-complete ?
\end{open}

\section{Substrate cycle}

Unlike to a tree, in a substrate cycle two paths connect two substrate nodes $u_i$ and $u_j$. Hence the algorithm for tree does not adapt for this case. 
In what follows, the path $u_i, u_{i+1}, \ldots, u_{j-1}, u_j$ (reps. $u_i, u_{i-1}, \ldots, u_{j+1}, u_j$) is said to be the \textit{clockwise} (resp. counterclockwise) path from $u_i$ to $u_j$, and is denoted $p^+(u_i, u_j)$ (resp. $p^-(u_i, u_j)$).

\subsection{Virtual path and cycle}

As a cycle is a Hamiltonian graph, from Proposition~\ref{proposition:section3:cyclepath:hamilto}, we have that E-$\langle \Path_r \rightarrow \Cycle_s \rangle$ and E-$\langle \Cycle_r \rightarrow \Cycle_s \rangle$ are in \p. 
We prove in the following that the cost variants is also in \p.

\paragraph{Virtual path}

A mapping of a virtual path $\Path_r$ on a substrate cycle $\Cycle_s$ is said to be an \textit{elementary path} if the concatenation of the routing of edges, $m_E(\ubar_1, \ubar_2) \cup m_E(\ubar_2, \ubar_3) \cup \ldots \cup m_E(\ubar_{n-1}, \ubar_n)$, is an elementary path of $\Cycle_s$, i.e. substrate nodes appear at most once in it. This case happens only when substrate edges host at most one virtual edge and (at least) one substrate edge hosts none (i.e. the mapping does not go around the whole cycle).
The following lemma shows that elementary path mappings are the structure of optimal solutions.

\vspace{0.25cm}
\begin{lemma} \label{lemma:section5:path:elementary}
    There is a minimum cost feasible mapping of $\Path_r$ on $\Cycle_s$ that is an elementary path.
\end{lemma}

\begin{proof}
    Let $m = (m_V, m_E)$ be a feasible mapping of $\Path_r$ on $\Cycle_s$ that is not an elementary path. We will construct an elementary path mapping $m' = (m'_V, m'_E)$ from $m$ of lower or equal cost, which will prove the lemma.
    
    We introduce a procedure that constructs an elementary clockwise path mapping, given the imput parameters $u_a \in V_s$ and $k:V_s \rightarrow \mathbb{N}$, such that, for $u_i \in V_s$, $k(u_i) \le c(u_i)$.
    The construction starts at a substrate node $u_a$ where it places the first $k(u_a)$ nodes. It then moves in a clockwise manner on the cycle, and places on each node $u_i$ the next $k(u_i)$ virtual nodes.
    The routing of $(\ubar_j, \ubar_{j+1})$, $j \in \{1, \ldots, n_r-1\}$, is set to the clockwise path between the placement of $\ubar_j$ and that of $\ubar_{j+1}$ (it is an empty path when they are placed on the same node). By construction, this mapping is an elementary path and is feasible.

    We now show that the cost of $m'$ is lower or equal than the one of $m$. The mapping $m' = (m'_V, m'_E)$ is obtained with parameters $k(u_i)$,  $u_i \in V_s$ set to the number of nodes hosted by $u_i$ in $m$, while the choice of $u_a$ depends on the configuration of $m$ as follows. If $m$ uses all substrate edges, $u_a$ can be set to any nodes that hosts some virtual nodes, say $u_1$. Otherwise, there is a node $u_j \in V_s$ such that, in $m$, $(u_{j-1}, u_j)$ is not used, and $(u_j, u_{j+1})$ is used. We set $u_a = u_j$. 
    
    The mapping $m'$ has the same node placement cost than $m$, since both mapping place the same number of virtual nodes on any substrate nodes. 
    It can also be shown that a substrate edge in $m'$ hosts at most the number of edges that were hosted in $m$. Hence the routing cost (and hence the overall cost) of $m'$ is lower or equal than the one of $m$. % It can be shown: I have it in another version, but we keep it short here to put emphasis on 5.3
\end{proof}

The following lemma proposes a polynomial algorithm to construct an elementary path mapping. 

\vspace{0.25cm}
\begin{lemma}\label{lemma:section5:path:greedy}
    For $u_s, u_t \in V_s$, a minimum cost feasible mapping that is a clockwise elementary path from $u_s$ to $u_t$ can be computed in $\BigO(n_s \log(n_s) + n_r)$
\end{lemma}

\begin{proof}
    We will use the procedure detailed in proof of Lemma \ref{lemma:section5:path:elementary} with $u_a$ set to $u_s$, and the parameters $k(u_i)$, $u_i \in V_s$, computed greedily as follows.
    First, all $k(u_i)$, $u_i \in V_s$, are set to zero, except for $k(u_s)$ and $k(u_t)$ which are set to one. Then, the nodes of $u_s, \ldots, u_t$ are sorted by increasing costs. Finally, substrate nodes are chosen greedily to host the remaining $n_r-2$ virtual nodes. A unit is iteratively added to $k(u_i)$ where $u_i$ is the cheapest node with some capacity left (i.e. such that $k(u_i) < c(u_i)$). 
    Overall, producing the mapping takes $\BigO(n_r + n_s log(n_s))$.
    
    The mapping obtained has minimum cost among those that are clockwise elementary paths from $u_s$ to $u_t$. Indeed, any mapping that is a clockwise elementary path from $u_s$ to $u_t$ has the same routing cost, using edges $(u_s, u_{s+1}), \ldots, (u_{t-1}, u_t)$ exactly once. However, the resulting mapping minimizes placement costs, as it uses, with respect to their capacities, the minimum cost substrate nodes of the subpath.
\end{proof}

Applying Lemma~\ref{lemma:section5:path:greedy}, for any $u_s, u_t \in V_s$, we obtain the final result:

\vspace{0.25cm}
\begin{corollary}\label{theorem:section5:path:cost}
    C-$\langle \Path_r \rightarrow \Cycle_s \rangle$ is in $\mathcal{P}$ and can be solved in $\BigO({n_s}^2 (n_s \log(n_s) + n_r))$
\end{corollary}
\vspace{0.25cm}

For a virtual tree, Lemma~\ref{lemma:section5:path:elementary} does not seem to be adaptable. The complexity of this case is an open problem.

\vspace{0.25cm}
\begin{open}\label{open:section5:path:tree}
    Is C-$\langle \Tree_r \rightarrow \Cycle_s \rangle$  \np-complete ?
\end{open}

\paragraph{Virtual cycle}

In the case of a virtual cycle $\Cycle_r$, there is a minimum cost feasible mapping of $\Cycle_r$ on $\Cycle_s$ where the partial mapping of the virtual subpath $\ubar_1, \ldots, \ubar_{n_r}$ is a clockwise elementary path. To show this, one can consider a mapping $m$ for which this is not true. The procedure of Lemma~\ref{lemma:section5:path:elementary} is used with input parameters derived as previously from $m$ to produce a partial mapping of the subpath $\ubar_1, \ldots, \ubar_{n_r}$, which is then completed by routing $(\ubar_1, \ubar_{n_r})$ as in $m$. The mapping obtained has a lower or equal cost.

Similarly, for $u_s, u_t \in V_s$, a minimum cost feasible mapping for which the partial mapping of $\ubar_1, \ldots, \ubar_{n_r}$ is an elementary path from $u_s$ to $u_t$ can be computed in polynomial time through the greedy algorithm of Lemma~\ref{lemma:section5:path:greedy}. Thus:

\vspace{0.25cm}
\begin{corollary} \label{theorem:section5:cycle:cost}
    C-$\langle \Cycle_r \rightarrow \Cycle_s \rangle$ is in \p and can be solved in $\BigO({n_s}^2 (n_s \log(n_s) + n_r))$.
\end{corollary}

\subsection{Virtual wheel}

We have shown that C-$\langle \Cycle_r \rightarrow \Cycle_s \rangle$ is in \p. It is known that C-$\langle \Star_r \rightarrow \Cycle_s \rangle$ is in \p \cite{rost2015stars}, through a reduction of the problem to a series of $n_s$ integer flow problem, which is solvable in polynomial time \cite{ahuja1988flow}.
In a graph $G = (V, E)$, let $P_{s,t}$ be the set of paths connecting $s, t \in V$. The problem can be defined as follows:  

\begin{decproblem}
    \problemtitle{\flowlong (\flow)}
    \probleminput{A graph $\Graph = (V, E)$, two nodes $s, t \in V$, capacities $c:E \rightarrow \mathbb{N}$ and cost $w:E \rightarrow \mathbb{N}$, integers $N, R \in \mathbb{N}$}
    \problemquestion{Is there a $N$-flow function $f : P_{s, t} \Rightarrow \mathbb{N}$ such that (a) $\sum_{p \in P_{s,t}} f(p) = N$, (b) for $e \in E$, $\sum_{p \in P_{s,t}, e \in p} f(p) \le c(e)$ and $(c)$ $\sum_{p \in P_{s,t}} \sum_{e \in p} w(e) f(p) \le R$ ?}
\end{decproblem}

In what follows, we will show that C-$\langle \Wheel_r \rightarrow \Cycle_s \rangle$ can be solved in polynomial time too, by combining ideas from the two algorithms. 

First, let us show that there is a minimum cost feasible mapping of $\Cycle_r$ on $\Cycle_s$ where the partial mapping of the virtual subpath $\ubar_1, \ldots, \ubar_{n_r}$ is a clockwise elementary path.
From a feasible mapping $m$ of $\Wheel_r$ on $\Cycle_s$, the procedure of the proof of Lemma~\ref{lemma:section5:path:elementary}, with input parameters derived from $m$ as previously, constructs a partial mapping $m'$ of the virtual subpath $\ubar_1, \ubar_2, \ldots, \ubar_{n_r}$ that is a clockwise elementary path. 
Then $m'$ can be completed as follows: $(\ubar_{n_r}, \ubar_1)$ is routed in the same direction as in $m$; the center of $\Wheel_r$ is placed on the same node as in $m$; the path that are used to connect $\ubar_c$ to nodes of the outer cycle hosted on a node $u_i \in V_s$ in $m$ are reused in $m'$. 
Note that, if the placement of a virtual node has changed in $m'$ w.r.t. $m$, then the path routing the corresponding virtual edge has to be changed accordingly. 
It can be verified that the mapping obtained is feasible and has a lower or equal cost than $m$.

Secondly, Lemma~\ref{lemma:section5:path:greedy} is also valid for a virtual wheel. However the previous greedy procedure does not work, as routing the extra virtual edges yields extra costs and possible capacities violation.
We propose another algorithm that solves a series of \flow, similarly to what is done by \cite{rost2015stars}.

\vspace{0.25cm}
\begin{lemma}\label{cycle:wheel:lemma:algo}
    For  $u_s, u_t, u_c \in V_s$, a minimum cost feasible mapping, such that $\ubar_c$ is placed on $u_c$ and the partial mapping of the virtual subpath $\ubar_1, \ldots, \ubar_{n_r}$ is a clockwise elementary path from $u_s$ to $u_t$, can be computed in polynomial time 
\end{lemma}

\begin{proof}
    We further consider the case where $(\ubar_{n_r}, \ubar_1)$ is routed counterclockwise, $(\ubar_1, \ubar_c)$ counterclockwise and $(\ubar_{n_r}, \ubar_c)$ clockwise. This is one of eight possible combinations for the routing direction of these edges. The other cases can be solved analogously. 
    We will show that the partial mapping given by the above node placement and edge routing can be completed from a solution of the following \flow instance.
    
    The graph $\Graph$ is constructed from $\Cycle_s$ with an additional node $S$. The original edges of $\Cycle_s$ in $\Graph$ have the same cost as in $\Cycle_s$. For capacities, the edges in $(u_a, u_{a+1}), \ldots, (u_{z-1}, u_z)$ have the same capacities as in $\Cycle_s$ minus two, and the other edges have the same capacities as in $\Cycle_s$ minus one (this is due to the configuration of routing of $(\ubar_{n_r}, \ubar_1)$, $(\ubar_1, \ubar_c)$ and $(\ubar_{n_r}, \ubar_c)$). 
    $S$ is connected to every node in $u_s, u_{s+1}, \ldots, u_t$. 
    The capacity of $(S, u_i)$ is the capacity of $u_i$, except for $(S, u_s)$ and $(S, u_t)$ where one unit of capacity is retrieved. 
    
    Consider a minimum-cost $(n_r-2)$-integer flow from $S$ to $u_c$. The partial mapping can be completed as follows. 
    The key idea is that, as done by \cite{rost2015stars}, the path of a unit of flow indicates both the node placement of a virtual node and the routing of the edge connecting this node to $\ubar_c$.
    However, contrary to the algorithm for a star, a unit of flow can not correspond to any virtual node. Indeed, the virtual nodes must be correctly ordered on the $u_s, \ldots, u_t$ such that the virtual subpath $\ubar_1, \ldots, \ubar_{n_r}$ is an elementary path. Hence, the units of flow are ranked by the first node after $S$ that they use: from the closest to $u_s$ to the furthest. 
    Consider the first path $(S, u_i, \ldots, u_c)$ obtained. Then $\ubar_2$ is placed on $u_i$, and $(\ubar_2, \ubar_c)$ is routed on  $(u_i, \ldots, u_c)$. The process is then repeated for $\ubar_3$ on the second path, etc., until all virtual nodes are placed. 
    Finally the edges $(\ubar_i, \ubar_{i+1})$, $i \in \{1, \ldots, n_r-1\}$, are routed clockwise.
    
    This mapping respects, by construction of the \flow instance, both node and edge capacities of $\Cycle_s$. Since the flow has minimum cost, the mapping is optimal for the case considered.
\end{proof}

Let $T_{\flow}(\Graph, N)$ be the time complexity of a (polynomial) \flow algorithm for finding a $N$-flow in a graph $\Graph$.
Applying Lemma~\ref{cycle:wheel:lemma:algo} for any $u_s, u_t, u_c \in V_s$, we obtain the result:

\vspace{0.25cm}
\begin{corollary}\label{cycle:wheel:theorem}
    C-$\langle W_r \rightarrow C_s \rangle $ is in $\mathcal{P}$, and solvable in $\BigO({n_s}^3 \times T_{\flow}(\Graph_s, n_r))$
\end{corollary}

\subsection{Virtual clique}

We propose a polynomial dynamic algorithm to solve \univne when the virtual network is a clique $\Clique_r$. 
First, as done for the dynamic algorithm of Section~\ref{section4:rest}, the key idea is to relax the knowledge of which nodes are placed on a subpath of $\Cycle_s$, and to consider instead the number of nodes placed on that subpath. 
In this more complex case of a substrate cycle, there may exist an exponential number of feasible virtual edge routings for a given node placement.
To tackle this exponentiality, we similarly relax the knowledge of the routing, considering instead the number of edges routed on a substrate edge. 
Let us now introduce some notations.

Let $W_V(u_i, k)$ (resp. $W_E(u_i, u_{i+1}, t)$) denote the part of the mapping cost induced by a substrate node $u_i$ (resp. edge $(u_i, u_{i+1}))$), $i \in \{1, \ldots, n_r\}$ when $k$ virtual nodes (resp. $t$ edges) are placed (resp. routed) on that node (resp. edge). In this trivial case, $W_V(u_i, k) = k \times w(u_i)$ (resp. $W_E(u_i, u_{i+1}, t) = t \times w(u_i, u_{i+1})$) when $k \le c(u_i)$ (resp. $t \le c(u_i, u_{i+1})$); and $W_V(u_i, k) = +\infty$ (resp. $W_E(u_i, u_{i+1}, t) = +\infty$) otherwise (as the mapping is unfeasible).

Let $W^*(u_i, k, a, b, c)$ denote the minimum mapping cost induced by a subpath $\{u_1, \ldots, u_i\}$, when $k$ virtual nodes are placed on the subpath, such that there are $a$, (resp. $b$, $c$) internal (resp. cut, external) edges of those $k$ nodes routed using $(u_i, u_{i+1})$.
Here, an edge is said to be internal (resp. cut, external) w.r.t. a node subset if it has both (resp. one, none) ends in the node subset.
Clearly, $W^*(u_1, k, a, b, c) = W_V(u_1, k)$. In the sequel, we deal with the recurrence formula for others $u_i$, $i \in \{1, \ldots, n_r\}$. \\

Consider a mapping and a substrate node $u_i$, $i \in \{1, \ldots, n_r-1\}$.
Let us denote $k$ (resp. $k'$) as the number of virtual nodes placed on the subpath $u_1, \ldots, u_{i}$ (resp.  $u_1, \ldots, u_{i+1}$) ; $A$, $B$ and $C$ (resp. $A'$, $B'$ and $C'$) the set of internal, cut and external edges of those $k$ (resp. $k'$) nodes ; and $a$,  $b$ and $c$ (resp. $a'$,  $b'$ and $c'$) the number of those internal, cut and external edges that use $(u_{i}, u_{i+1})$ (resp. $(u_{i+1}, u_{i+2})$). Let $\Phi(k', a', b', c')$ denote the set of possible values that $k$, $a$, $b$ and $c$ might hold for given $k', a', b'$ and $c'$.
% $\hat{a}'$  $\hat{b}'$ and $\hat{c}'$ (resp. $\hat{a}$  $\hat{b}$ and $\hat{c}$) the number of those internal, cut and external edges that do not use $(u_{i+1}, u_{i+2})$ (resp. $(u_i, u_{i+1})$).

\begin{figure}
\centering
\begin{tikzpicture}
\begin{scope} [xshift = -3.5cm]
    
        \node[vnode, minimum size=0.4cm] (v1) at (1.5*360/4:1.4){$\ubar_1$};
        \node[vnode, minimum size=0.4cm] (v2) at (0.5*360/4:1.4){$\ubar_2$};
        \node[vnode, minimum size=0.4cm] (v3) at (3.5*360/4:1.4){$\ubar_3$};
        \node[vnode, minimum size=0.4cm] (v4) at (2.5*360/4:1.4){$\ubar_4$};
        \draw[networkedge, magenta] (v1) -- (v2);
        \draw[networkedge, blue] (v2) -- (v3);
        \draw[networkedge, red] (v3) -- (v4);
        \draw[networkedge, orange] (v4) -- (v1);
        \draw[networkedge, green] (v1) -- (v3);
        \draw[networkedge, yellow] (v2) -- (v4);

\end{scope}

\begin{scope} [xshift = 3.5cm]
    % substrate cycle

    \node[snode] (u1) at (-2.5, 0.8){$u_1$};
    \node[snode] (u2) at (0, 1.5){$u_2$};
    \node[snode] (u3) at (2.5, 0.8){$u_3$};
    \node[snode] (u4) at (1.75, -1.25){$u_4$};
    \node[snode] (u5) at (-1.75, -1.25){$u_5$};
    \draw[networkedge] (u1) -- (u2);
    \draw[networkedge] (u2) -- (u3);
    \draw[networkedge] (u3) -- (u4);
    \draw[networkedge] (u4) -- (u5);
    \draw[networkedge] (u5) -- (u1);

    % placement des noeuds virtuels
    \node[vnode, minimum size=0.4cm] (v1) at (-3.05, 1.4){$\ubar_1$};
    \node[vnode, minimum size=0.4cm] (v2) at (0, 2.35){$\ubar_2$};
    \node[vnode, minimum size=0.4cm] (v3) at (3.05, 1.40){$\ubar_3$};
    \node[vnode, minimum size=0.4cm] (v4) at (2.35, -1.85){$\ubar_4$};

    % routages
     \draw[networkedge, magenta] (v1) to [bend left = 15] (u1);
     \draw[networkedge, magenta] (u1) to [bend left = 15] (u2);
     \draw[networkedge, magenta] (u2) to [bend left = 15] (v2);

     \draw[networkedge, blue] (v2) to [bend left = 15] (u2);
     \draw[networkedge, blue] (u2) to [bend left = 20] (u3);
     \draw[networkedge, blue] (u3) to [bend left = 20] (v3);
     
     \draw[networkedge, red] (v3) to [bend left = 20] (u3);
     \draw[networkedge, red] (u3) to [bend left = 20] (u4);
     \draw[networkedge, red] (u4) to [bend left = 20] (v4);
     
     \draw[networkedge, orange] (v4) to [bend left = 15] (u4);
     \draw[networkedge, orange] (u4) to [bend left = 15] (u5);
     \draw[networkedge, orange] (u5) to [bend left = 15] (u1);
     \draw[networkedge, orange] (u1) to [bend left = 15] (v1);
     
     \draw[networkedge, green] (v1) to [bend left = 0] (u1);
     \draw[networkedge, green] (u1) to [bend right = 7] (u5);
     \draw[networkedge, green] (u5) to [bend right = 7] (u4);
     \draw[networkedge, green] (u4) to [bend right = 7] (u3);
     \draw[networkedge, green] (u3) to [bend right = 0] (v3);
     
     \draw[networkedge, yellow] (v2) to [bend left = 0] (u2);
     \draw[networkedge, yellow] (u2) to [bend left = 13] (u3);
     \draw[networkedge, yellow] (u3) to [bend left = 13] (u4);
     \draw[networkedge, yellow] (u4) to [bend left = 0] (v4);
\end{scope}

\end{tikzpicture}
\caption{Example of a mapping of $\Clique_4$ on $\Cycle_5$}
\label{cycle:clique:fig}
\end{figure}

Figure \ref{cycle:clique:fig} provides an example of a virtual clique $\Clique_4$ on the left, and its mapping on a cycle $\Cycle_5$ on the right.
The mapping is the following: a virtual node is shown next to the substrate node it is placed on, and the routing of a virtual edge is shown on the cycle with its corresponding color.
For $u_i = u_2$, we have
$k = 2$;
$A = \{(\ubar_1, \ubar_2)\}$, $B = \{(\ubar_1, \ubar_3), (\ubar_2, \ubar_3), (\ubar_1, \ubar_4), (\ubar_2, \ubar_4)\}$ and $C = \{(\ubar_3, \ubar_4)\}$;
$a = 0$, $b = 2$ and $c = 0$.
And for $u_{i+1} = u_3$, we have 
$k' = 3$; 
$A' = \{(\ubar_1, \ubar_2), (\ubar_1, \ubar_3), (\ubar_2, \ubar_3)\}$, $B' = \{(\ubar_1, \ubar_4), (\ubar_2, \ubar_4), (\ubar_3, \ubar_4)\}$ and $C' = \{\emptyset\}$;
$a' = 1$, $b' = 2$ and $c' = 0$.

\vspace{0.25cm}
\begin{lemma}\label{cycle:clique:lemmarec}
    For $i \in \{1, \ldots, n_r-1\}$, the following recurrence formula holds:
    \begin{equation*}  \label{recurrence}
    \begin{split}
        W^*(u_{i+1}, k', a', b, c) =
        \min_{k, a, b, c \in \Phi(k', a', b', c')} &   W^*(u_i, k, a, b, c) \\
        & + W_E(u_i, u_{i+1}, a, b, c) + W_V(u_{i+1}, k'-k)
    \end{split}
    \end{equation*}
\end{lemma}

\begin{proof}
    Consider a mapping $m$ that has minimum cost $W^*(u_{i+1}, k', a', b', c')$ on $\{u_1, \ldots, u_{i+1}\}$, but does not have minimum cost $W^*(u_i, k, a, b, c)$ on $\{u_1, \ldots, u_{i}\}$. Then there is a mapping $m'$ with this minimum cost. We assume that the nodes are placed by $m$ and $m'$ on $u_1, \ldots, u_i$ are the same, w.l.o.g. since all nodes are equivalent. Let us construct a new mapping $m^*$, obtained from $m$ and $m'$ as follows. 
    
    The placement of a virtual node $\ubar_j$, if it is one of the $k$ (resp. $n_r - k$) nodes placed on $u_1, \ldots, u_{i}$ (resp. on $u_{i+1}, \ldots, u_{n_s}$), is its placement in $m'$ (resp. $m$): $m^*_V(\ubar_j) = m'_V(\ubar_j)$ (resp. $m^*_V(\ubar_j) = m_V(\ubar_j)$).
    
    The routing of an internal (resp. external) edge $(\ubar_{j_1}, \ubar_{j_2})$ is its routing in $m'$ (resp. $m$): $m_E^*(\ubar_{j_1}, \ubar_{j_2}) = m_E'(\ubar_{j_1}, \ubar_{j_2})$ (resp. $m_E^*(\ubar_{j_1}, \ubar_{j_2}) = m_E(\ubar_{j_1}, \ubar_{j_2})$).
    
    The routing of a cut edge $(\ubar_{j_1}, \ubar_{j_2})$ uses parts of paths from $m_E$ and $m'_E$ as follows.
    If the routing path of the edge uses $(u_i, u_{i+1})$ (resp. $(u_{n_s}, u_1)$) in both mapping, then the new mapping is simply a concatenation from parts of the two previous paths: $m^*_E(\ubar_{j_1}, \ubar_{j_2}) = p^+(m'_V(\ubar_{j_1}), u_i) \cup p^+(u_i, m_V(\ubar_{j_2}))$ (resp. $m^*_E(\ubar_{j_1}, \ubar_{j_2}) = p^-(m'_V(\ubar_{j_1}), u_1) \cup p^-(u_1, m_V(\ubar_{j_2}))$).
    
    Otherwise, several cases are possible. Consider for instance that the routing of $(\ubar_{j_1}, \ubar_{j_2})$ uses $(u_i, u_{i+1})$ in $m$ and $(u_{n_s}, u_1)$ in $m'$, then there is another virtual edge, say $(\ubar_{j'_1}, \ubar_{j'_2})$, for which the opposite is true. Indeed, there are $b$ (resp. $k \times (n_r - k) - b$) cut edges that are routed using $(u_i, u_{i+1})$ (resp. $(u_{n_s}, u_1)$) in both mappings.
    This more complicated case can be tackled as follows. The node placement of $\ubar_{j_2}$ and $\ubar_{j'_2}$ are swapped in $m^*$. % (but the routing paths already set do not change). 
    This is possible since the two nodes are rigorously equivalent. Then the routing paths in $m^*$ combines parts of previous routing paths as follows:
    \begin{align*}
        & m^*_E(\ubar_{j_1}, \ubar_{j_2}) && = && p^-(m'_V(\ubar_{j_1}), u_1) \cup p^-(u_1, m_V(\ubar_{j'_2})) \\
        & m^*_E(\ubar_{j'_1}, \ubar_{j'_2}) && = && p^+(m'_V(\ubar_{j'_1}), u_i) \cup p^+(u_i, m_V(\ubar_{j_2}))
    \end{align*}
    This operation can be reapplied on other such cut edges, until the edge routing of $m^*$ is complete.
    
    We now show that $m^*$ is feasible and of better cost. 
    In that aim, let us show that the substrate nodes and edges utilization, over $u_1, \ldots, u_i$, are the same in $m^*$ as in $m'$. 
    For substrate nodes, this is trivial since the node placement is the same.
    For substrate edges, note that the routing of internal and cut edges over that segment comes from $m'$. Moreover, although the routing of external edges comes from $m$, observe that for both mappings $m$ and $m'$ there are $c$ external edges whose routing use every substrate edge of $u_1, \ldots, u_i$ (the other external edges use none of those). Thus on that segment the edge utilization is the same as in $m'$.
    
    Similarly, it can be shown that the substrate nodes and edges utilization on $u_{i+1}, \ldots, u_{n_r}$ are the same in $m^*$ as in $m$. 
    Finally, for edges $(u_i, u_{i+1})$ and $(u_{n_r}, u_1)$, the utilization of $m$, $m'$ and $m^*$ are the same.
    Since $m$ and $m'$ are feasible, this shows that $m^*$ is feasible. 
    
    This also shows that the cost of $m^*$ induced by $\{u_1, \ldots, u_i\}$ is the same as $m'$, i.e. $W^*(u_i, k, a, b, c)$. Since the utilization of $(u_i, u_{i+1})$ and $u_{i+1}$ is the same in $m$ and in $m^*$, the cost of $m^*$ on $\{u_1, \ldots, u_{i+1}\}$ is better than $W^*(u_{i+1}, k', a', b', c')$, which is a contradiction.
\end{proof}

The following lemma details how to compute the "$\min$" term of the previous formula.

\vspace{0.25cm}
\begin{lemma} \label{cycle:clique:lemmaphi}
    For given values of $k'$, $a'$, $b'$ and $c'$, the size of $\Phi(k', a', b', c')$ is in $\BigO({n_r}^5)$.
\end{lemma}

\begin{proof}
    Let us note $l = k' - k$. Then:
    \begin{itemize}
        \item There are $\frac{l (l-1)}{2}$ edges of $A'$ that are in $C$: those are the edges that connect nodes that are both placed on $u_{i+1}$, thus do not use $(u_{i}, u_{i+1})$.
        \item There are $kl$ edges of $A'$ that are in $B$. 
        Let us denote  $\gamma$ the number of those edges which use $(u_i, u_{i+1})$. 
        Then $kl \ge \gamma \ge 0$
        % or : Then $kl \ge \gamma \ge \max(0, kl - (\frac{k'(k'-1)}{2} -a'))$. % or :         Then $kl \ge \gamma \ge 0$.
        Note that those $\gamma$ edges do not use $(u_{i+1}, u_{i+2})$.
        \item There are $(n_r-(k+l))l$ edges of $B'$ that are in $C$.
        Let us denote $\delta$ the number of those edges which use $(u_i, u_{i+1})$. Then $(n_r-(k+l))l \ge \delta \ge 0$.  % or Then $(n_r-(k+l))l \ge \delta \ge \max(0,  (n_r-(k+l))l - \hat{b}')$ ???.
        Note that those $\delta$ edges do not use $(u_{i+1}, u_{i+2})$. % hence $(n_r-(k+l))l - \delta$ edges are added to $b'$.
        \item Finally, all the edges of $A'$ remain in $A$.
    \end{itemize}
    
    This allows us to derive the following equations: $c = c' + \delta$, $b = b'- ((n_r-(k+l))l - \delta) + \gamma$ and $a = a' - (kl - \gamma)$.
    
    The observation that $l$ can take at most $n_r$ different values, and both $\delta$ and $\gamma$ can take at most ${n_r}^2$ different values completes the proof.
\end{proof}

Finally, we obtain the complexity of the polynomial dynamic algorithm resulting from the formula of Lemma \ref{cycle:clique:lemmarec}.%

\vspace{0.25cm}
\begin{theorem} \label{cycle:clique:theorem}
    C-$\langle \Clique_r \rightarrow \Cycle_s \rangle$ is in \p and can be solved in $\BigO({n_r}^{12} n_s)$
\end{theorem}

\begin{proof}
    The values of $W^*$ are computed iteratively on the cycle, starting from $u_1$, using  formula of Lemma \ref{cycle:clique:lemmarec}.
    At each node, there are $\BigO({n_r}^7)$ values of $W^*$ to be computed. Indeed, the possible values in $k'$ are in $\BigO(n_r)$, and those of $a'$ (resp. $b'$, $c'$) are in $\BigO({n_r}^2)$, since there are $\frac{k (k-1)}{2}$ (resp. $k(n_r -k)$,  $\frac{(n_r - k)(n_r - k - 1)}{2}$) internal (resp. cut, external) edges.
    
    Moreover, from Lemma \ref{cycle:clique:lemmaphi}, each value of $W^*$ is computed in $\BigO({n_r}^5)$.
    The algorithm finds all values of $W^*(u_{n_s}, n_r, a, 0, 0)$, $a \in \{0, \ldots, |E_r|\}$, in  $\BigO({n_r}^{12} n_s)$. The last operation is to add the edge value $W_E(u_{n_s}, u_1, a)$ and compute the minimum value, which takes $\BigO({n_r}^2)$.
\end{proof}

The algorithm is clearly impracticable. The simple fact that there is two paths between two nodes in a cycle makes the dynamic algorithm approach much more expensive than on a tree. Note that this approach can be further used for any virtual topology on a substrate cycle, but the algorithm in the general case becomes exponential. In some cases, such as paths or wheels, it might be adapted to be polynomial, but at the cost of high computation complexity. Hence the hardness of the problem in the general case remain an open problem.

\vspace{0.25cm}
\begin{open}\label{cycle:clique:open}
    Is E-$\langle \Graph_r \rightarrow \Cycle_s \rangle$  \np-complete ?
\end{open}
\vspace{0.25cm}

Although we have proven several polynomial cases above, we conjecture that the problem is \np-complete for a general virtual graph on a substrate cycle, since the hardness of \univne on a cycle can be expected to be at least as hard as on a substrate path.

\section{Conclusion}

The results of this article, together with existing results, are synthesized in Table~\ref{table-complexity}, where a green (resp. yellow, red) cell indicates that the case is solvable in polynomial time (resp. is an open problem, is \np-complete). 
When the results differ over the variant, the corresponding cell is split with the left (resp. right) half representing the E-\univne (resp. C-\univne) result. An arrow indicates when a result is implied by an other.

\newcolumntype{C}[1]{>{\centering\let\newline\\\arraybackslash\hspace{0pt}}m{#1}}
\begin{table}[h!]
    \centering
    \renewcommand{\arraystretch}{1.5} % Adjust row height (default is 1.0)
    \begin{tabular}{C{8pt}|C{45pt}|C{36pt}|C{36pt}|C{36pt}|C{36pt}|C{30pt}|C{30pt}|C{36pt}|}
    
    \multicolumn{2}{c}{}&\multicolumn{7}{c}{\bf Substrate}\\
    \hhline{~~-------}
    \multicolumn{2}{c|}{}& Star & Tree &  Path & Cycle & \multicolumn{2}{C{60pt}|}{Clique} & Gen. \\ \hhline{~--------} 

    \multirow{7}[0]{*}{\begin{sideways}\bf Virtual ~~~\end{sideways}}  & Star & \greencell $\rightarrow$ & \greencell (\cite{rost2015stars}) &\greencell $\leftarrow$ & \greencell (\cite{rost2015stars})  & \multicolumn{2}{c|}{\greencell $\rightarrow$} & \greencell (\cite{rost2015stars}) \\ \hhline{~--------} 

    & Tree & \redcell Thm~\ref{theorem:section4:tree:exist} & \redcell Thm~\ref{theorem:section4:tree:exist2} & \yellowcell OP~\ref{theorem:section4:gen:open}  &  \yellowcell OP~\ref{open:section5:path:tree}  & \yellowcell OP~\ref{theorem:section3:wheeltree:open} & \redcell $\downarrow$ & \redcell $\downarrow$ \\ \hhline{~--------} 

    & Path &  \greencell $\rightarrow$ & \greencell Thm~\ref{theorem:section4:rest:pathcost} & \greencell $\leftarrow$ & \greencell Thm~\ref{theorem:section5:path:cost} & \greencell Prop~\ref{proposition:section3:cyclepath:hamilto} & \redcell Thm~\ref{theorem:section3:cyclepath:cost}  & \redcell (\cite{wu2020path}) \\ \hhline{~--------} 

    & Cycle & \greencell $\rightarrow$ & \greencell Thm~\ref{theorem:section4:cycle:cost} & \greencell $\leftarrow$ & \greencell Thm~\ref{theorem:section5:cycle:cost} & \greencell Prop~\ref{proposition:section3:cyclepath:hamilto} & \redcell Thm~\ref{theorem:section3:cyclepath:cost}  & \redcell Thm~\ref{E-CrGs}  \\ \hhline{~--------} 
    
    & Wheel & \greencell $\rightarrow$ & \greencell Thm~\ref{theorem:section4:rest:wheelcost} & \greencell $\leftarrow$ & \greencell Thm~\ref{cycle:wheel:theorem} & \yellowcell OP~\ref{theorem:section3:wheeltree:open} & \redcell Thm~\ref{theorem:section3:wheeltree:cost}& \redcell Cor~\ref{corollary:section2:cyclewheel:E-WrGs}  \\ \hhline{~--------} 
    
    & Clique & \greencell $\rightarrow$ & \greencell Thm~\ref{theorem:section4:rest:cliquecost} & \greencell $\leftarrow$ & \greencell Thm~\ref{cycle:clique:theorem} & \multicolumn{2}{c|}{\redcell Thm~\ref{theorem:section3:clique:exist} } & \redcell Thm~\ref{theorem:section2:cyclewheel:E-KrGs} \\ \hhline{~--------} 

    & Gen. & \redcell $\uparrow$ & \redcell $\leftarrow$ $\rightarrow$ & \redcell Thm~\ref{theorem:section4:gen:exist} & \yellowcell OP~\ref{cycle:clique:open}  & \multicolumn{2}{c|}{ \redcell $\uparrow$ } & \redcell $\uparrow$ \\ \hhline{~--------} 

    \end{tabular}
    \caption{Complexity of \univne}
    \label{table-complexity}
\end{table}

We have proposed a large study of the frontier between polynomial and \np-hard cases of \univne. These complexity results give more insights about the computational challenges of \vne. 
Since the considered network topologies correspond to real use-cases, our polynomial algorithms can be directly used in practice, particularly those on trees, such as in data-centers \cite{ballani2011datacenter}.
They are directly linked to practical applications as the considered topologies correspond to networks encountered in practice.

The problems that are left opened raise several questions on intriguing cases, like E-$\langle \Wheel_r \rightarrow \Clique_s \rangle$. Other topologies may also be of significant interest for further study, like star of stars (called 2-stars by \cite{pankratov2023tree}) or cactuses.

Another important direction is how the results evolve when several virtual graphs have to be mapped simultaneously on the substrate graph. Dynamic programming algorithms on substrate trees, for instance, might be extended for embedding several virtual paths. 

Implementing real virtualized networks requires further dedicated optimization methods. Among possible approaches, Mathematical Programming has been shown to be efficient for solving network design problems \cite{kerivin2005design, benhamiche2020capacitated}.
To this aim, finding polyhedral characterizations for the polynomial cases we have proposed could help to improve branch-and-cut approaches. Another perspective is to use decomposition schemes for large integer programs for which the pricing subproblems can be solved in polynomial time.

\bibliographystyle{abbrv}
\bibliography{bibliography}

\end{document}